\documentclass[%
  reprint,
]{revtex4-2}

\usepackage{amsmath}
\usepackage{amssymb}
\usepackage{bm,bbm}

\usepackage{graphicx}
\usepackage[bookmarks=false]{hyperref}
\hypersetup{colorlinks=true,citecolor=blue,linkcolor=red,%
urlcolor=blue,pdfstartview=FitH,bookmarksopen=true}
\usepackage[T1]{fontenc}
\usepackage[osf,sc]{mathpazo}
\usepackage{dsfont}

\usepackage{amsthm}
\newtheoremstyle{note}% Name
{3pt}% Space above
{3pt}% Space below
{}% Body font
{}% Indent amount
{\itshape}% Theorem head font
{:}% Punctuation after theorem heading
{.5em}% Space after theorem heading
{}% Theorem head spec (can be left empty, meaning `normal')

\newtheorem{theorem}{Theorem}

\newtheorem{lemma}[theorem]{Lemma}

\DeclareMathOperator{\Tr}{Tr}
\DeclareMathOperator{\Var}{Var}

\DeclareMathOperator{\rank}{rank}

\newcommand{\bra}[1]{\langle #1\rvert}
\newcommand{\ket}[1]{\lvert #1\rangle}

\newcommand{\braket}[2]{\langle #1\vert #2\rangle}
\newcommand{\abs}[1]{\lvert #1\rvert}
\newcommand{\norm}[1]{\lVert #1\rVert}
\newcommand{\vect}[1]{\bm{#1}}
\newcommand{\md}{\mathrm{d}}
\newcommand{\me}{\mathrm{e}}

\newcommand{\Cl}{\mathrm{cl}}

\newcommand{\I}{\mathds{1}}

\newcommand{\maxover}[1][]{\underset{#1}{\mathrm{max}}}
\newcommand{\minover}[1][]{\underset{#1}{\mathrm{min}}}
\newcommand{\subto}{\mathrm{subject~to}}

\newcommand{\Sep}{\mathrm{SEP}}

\newcommand{\vm}{\bm{m}}
\newcommand{\vp}{\bm{p}}

\newcommand{\vv}{\bm{v}}
\newcommand{\vu}{\bm{u}}
\newcommand{\vx}{\bm{x}}
\newcommand{\vy}{\bm{y}}

\newcommand{\cH}{\mathcal{H}}

\newcommand{\cM}{\mathcal{M}}
\newcommand{\cN}{\mathcal{N}}
\newcommand{\cO}{\mathcal{O}}

\newcommand{\cT}{\mathcal{T}}

\newcommand{\dN}{\mathds{N}}
\newcommand{\dR}{\mathds{R}}

\newcommand{\bmid}{\;\big|\;}

\begin{document}

\title{Optimal Entanglement Certification from Moments of the Partial Transpose}

\author{Xiao-Dong Yu}
% \email{Xiao-Dong.Yu@uni-siegen.de}
\affiliation{Naturwissenschaftlich-Technische Fakult\"at, Universit\"at Siegen,
Walter-Flex-Stra{\ss}e 3, D-57068 Siegen, Germany}

\author{Satoya Imai}
\affiliation{Naturwissenschaftlich-Technische Fakult\"at, Universit\"at Siegen,
Walter-Flex-Stra{\ss}e 3, D-57068 Siegen, Germany}

\author{Otfried G\"uhne}
% \email{otfried.guehne@uni-siegen.de}
\affiliation{Naturwissenschaftlich-Technische Fakult\"at, Universit\"at Siegen,
Walter-Flex-Stra{\ss}e 3, D-57068 Siegen, Germany}

\date{\today}

\begin{abstract}
 For the certification and benchmarking of medium-size quantum devices 
 efficient methods to characterize entanglement are needed. In this context, it 
 has been shown that locally randomized measurements on a multiparticle quantum 
 system can be used to obtain valuable information on the so-called moments of 
 the partially transposed quantum state. This allows one to infer some 
 separability properties of a state, but how to use the given information in an 
 optimal and systematic manner has yet to be determined. We propose two general 
 entanglement detection methods based on the moments of the partially 
 transposed density matrix. The first method is based on the Hankel matrices 
 and provides a family of entanglement criteria, of which the lowest order 
 reduces to the known $p_3$-PPT criterion proposed in A. Elben \textit{et al.}, 
 [Phys.~Rev.~Lett. 125, 200501 (2020)]. The second method is optimal and gives 
 necessary and sufficient conditions for entanglement based on some moments of 
 the partially transposed density matrix.
\end{abstract}

\maketitle

%%%%%%%%%%%%%%%%%%%%%%%%%%%%%%%%%%%%%%%%%%%%%%%%%%%%%%%%%%%%%%%%%%%%%%%%%%%%%%
% \section{Introduction}
\textit{Introduction.---}%
%%%%%%%%%%%%%%%%%%%%%%%%%%%%%%%%%%%%%%%%%%%%%%%%%%%%%%%%%%%%%%%%%%%%%%%%%%%%%%
%
Intermediate-scale quantum devices involving a few dozen qubits  are  
considered a stepping stone toward the ultimate goal of achieving 
fault-tolerant quantum computation \cite{Preskill2018}. For such devices, the 
standard method of tomography is no longer feasible for gauging the performance 
in actual experiments \cite{Paris.Rehacek2004}. As a result, efficient and 
reliable characterization methods of such multiparticle systems are 
indispensable for current quantum information research 
\cite{Eisert.etal2020,Kliesch.Roth2021}. As entanglement is a key ingredient in 
quantum computation and other quantum information processing tasks, many 
efforts have been devoted to its characterization and quantification 
\cite{Horodecki.etal2009, Guehne.Toth2009, Friis.etal2019}.

If an experiment aims at producing a specific quantum state with few particles, 
entanglement witnesses or Bell inequalities provide mature tools for 
entanglement detection. For larger and noisy systems, however, these methods 
require significant measurement efforts; moreover, some of the standard 
constructions of witnesses are not very powerful. To overcome this, methods 
using locally randomized measurements have been put forward.  In these schemes, 
one performs on the particles measurements in random bases and determines the 
moments from the resulting probability distribution. It was noted early that 
this approach allows one to detect entanglement 
\cite{Tran.etal2015,Tran.etal2016} or to evaluate the moments of the density 
matrix \cite{vanEnk.Beenakker2012}. Recently, this approach turned into the 
center of attention and found experimental applications. For instance, it was 
shown that with these methods entropies can be estimated 
\cite{Elben.etal2019,Brydges.etal2019}, different forms of multiparticle 
entanglement can be characterized 
\cite{Ketterer.etal2019,Ketterer.etal2020,Ketterer.etal2020b,Knips.etal2020}, 
and also bound entanglement as a weak form of entanglement can be detected 
\cite{Imai.etal2021}. Especially, many efforts have been devoted to verify the 
positive partial transpose (PPT) condition \cite{Peres1996} from the moments of 
the randomized measurements \cite{Gray.etal2018,Elben.etal2020,Zhou.etal2020}.  

To explain this approach, let $\rho_{AB}$ be a quantum state in a bipartite 
quantum system $\cH_A\otimes\cH_B$, then the PPT criterion states that for any 
separable state $\rho_{AB}^{T_A}\ge 0$, where $^{T_A}$ denotes the partial 
transposition on subsystem $\cH_A$.  For a given quantum state $\rho_{AB}$, it 
is straightforward to check whether the PPT criterion is violated, and if so, 
the state must be entangled. However, in actual experiments, the quantum state 
is unknown, unless the resource-inefficient quantum state tomography is 
performed.  Recently, researchers found that the PPT condition can also be 
studied by considering the so-called partial transpose moments (PT-moments)
\begin{equation}
  p_k:=\Tr\Big[ \big( \rho_{AB}^{T_A} \big)^k \Big],
  \label{eq:PTmoment}
\end{equation}
which can be efficiently measured from randomized measurements 
\cite{Elben.etal2020,Huang.etal2020}.  To see the basic idea behind the 
PT-moment-based entanglement detection, suppose that we know all the PT-moments 
$\vp=(p_0,p_1,p_2,\dots,p_d)$, where $d=d_Ad_B$ is the dimension of the global 
system $\cH_A\otimes\cH_B$. Then, all the eigenvalues of $\rho_{AB}^{T_A}$ can 
be directly calculated \cite{Ekert.etal2002}, from which we can verify whether 
the PPT criterion is violated. Hereafter, we always assume that $p_0=d$ and 
$p_1=1$, which are trivial, but included for convenience.

In practice, however, it is difficult, if not impossible, to measure all the 
moments of a quantum state. Hence, the problem turns to whether we can detect 
the entanglement from the moments of limited order.  The question whether 
knowledge of some moments allows to draw conclusions about the underlying 
probability distribution is indeed fundamental, and has appeared in quantum 
information theory before 
\cite{BohnetWaldraff.etal2017,Milazzo.etal2020,Shchukin.Vogel2005}. For the 
case of PT-moments, we formulate this problem as follows:

\noindent\textbf{PT-Moment Problem.} Given the PT-moments of order $n$, is 
there a separable state compatible with the data? More technically formulated, 
given the numbers  $\vp^{(n)}=(p_0,p_1,p_2,\dots,p_n)$, is there a separable 
state $\rho_{AB}$ such that $p_k=\Tr[(\rho_{AB}^{T_A})^k]$ for $k=0,1,\dots,n$?

In Ref.~\cite{Elben.etal2020} the PT-moment problem for $n=3$ was studied and 
a necessary (but not sufficient) condition was proposed, called the $p_3$-PPT 
criterion:
\begin{equation}
  \rho_{AB}\in\Sep~\Rightarrow~
  p_3\ge p_2^2,
  \label{eq:p3PPT}
\end{equation}
where $\Sep$ denotes the set of separable states.

In this work, we propose two systematic methods for solving the general 
PT-moment problem.  First, we build a connection between the PT-moment problem 
and the known moment problems in the mathematical literature.  This gives 
a relaxation of the PT-moment problem, resulting in a family of entanglement 
criteria, in which the $p_3$-PPT criterion is the lowest order. Second, we show 
that the $p_3$-PPT criterion is not sufficient for the PT-moment problem of 
order three.  By reformulating the PT-moment problem as an optimization 
problem, we derive an explicit necessary and sufficient criterion for $n=3$ and 
further generalize it to the case that $n>3$. At last, we illustrate the 
efficiency of our criteria with physically relevant examples, e.g., thermal 
states of condensed matter systems.

%%%%%%%%%%%%%%%%%%%%%%%%%%%%%%%%%%%%%%%%%%%%%%%%%%%%%%%%%%%%%%%%%%%%%%%%%%%%%%
% \section{Relaxation to the classical moment problems}
\textit{Relaxation to the classical moment problems.---}%
%%%%%%%%%%%%%%%%%%%%%%%%%%%%%%%%%%%%%%%%%%%%%%%%%%%%%%%%%%%%%%%%%%%%%%%%%%%%%%
%
We start by relaxing the PT-moment problem and establishing a connection to the 
classical moment problems.  Here, instead of defining the classical moment 
problems with respect to the Borel measure on the real 
line~\cite{Akhiezer1965,Lasserre2010,Schmuedgen2017}, we rewrite them with 
quantum states and observables.

Given a quantum state $\sigma$ and an observable (Hermitian operator) $X$, the 
$k$-th moment is defined as
$m_k:=\Tr\big(\sigma X^k\big)$.
The moment problems ask the converse: given a sequence of moments, does there 
exist a quantum state $\sigma$ and an observable $X$ (with some restrictions) 
giving the desired moments? Albeit formulated in a quantum language, this 
scenario is essentially classical, since $\sigma$ can be taken diagonal in the 
eigenbasis of $X$.  Especially, the (truncated) Hamburger and Stieltjes moment 
problems are defined as follows:

\noindent\textbf{Hamburger Moment Problem.} Given the moments of order $n$, 
more precisely, $\vm^{(n)}=(m_0, m_1,m_2,\dots,m_n)$, is there a quantum state 
$\sigma$ and an observable $X$ such that $m_k=\Tr(\sigma X^k)$ for 
$k=0,1,\dots,n$?

\noindent\textbf{Stieltjes Moment Problem.} Given the moments of order $n$, 
more precisely, $\vm^{(n)}=(m_0, m_1,m_2,\dots,m_n)$, is there a quantum state 
$\sigma$ and a positive semidefinite observable $X$ such that $m_k=\Tr(\sigma 
X^k)$ for $k=0,1,\dots,n$?

Clearly, the only difference between these problems is that in the Stieltjes 
moment problem $X$ has to be positive semidefinite. We define the corresponding 
two sets of moments as
\begin{align}
  \cM_n&={\Big\{\vm^{(n)}\mid\Tr(\sigma X^k)=m_k,
  ~\sigma\ge 0,~X^\dagger=X\Big\}},
  \label{eq:momentSpace}\\
  \cM_n^+&={\Big\{\vm^{(n)}\mid\Tr(\sigma X^k)=m_k,
  ~\sigma\ge 0,~X\ge 0\Big\}}.
  \label{eq:momentSpacePSD}
\end{align}
Note that in the above definitions there is no restriction on the dimension of 
$\sigma$ and $X$. Also, since there is no bound on the eigenvalues of $X$, the 
sets $\cM_n$ and $\cM_n^+$ are not closed.

If we set $\sigma=\I$ and $X=\rho_{AB}^{T_A}$, the PT-moments 
$\vp^{(n)}=(p_0,p_1,\dots,p_n)$ defined by Eq.~\eqref{eq:PTmoment} always 
satisfy that $\vp^{(n)}\in\cM_n$, and furthermore the PT-moments given by the 
PPT states satisfy that $\vp^{(n)}\in\cM_n^+$. Hence, if we can characterize 
the set $\cM_n^+$, or the difference between $\cM_n^+$ and 
$\cM_n\setminus\cM_n^+$, we get a family of necessary conditions for the 
PT-moment problem. This is a relaxation, as in the definition of $\cM_n$ and 
$\cM_n^+$ more general $\sigma$ are allowed.

To proceed, we introduce the notion of Hankel matrices. The Hankel matrices 
$H_k(\vm)$ and $B_k(\vm)$ are $(k+1)\times (k+1)$ matrices defined by
\begin{equation}
  [H_k(\vm)]_{ij}=m_{i+j},~
  [B_k(\vm)]_{ij}=m_{i+j+1},~
  \label{eq:Hankel}
\end{equation}
for $i,j=0,1,\dots,k$. Hereafter, we will often suppress the argument ($\vm$ or 
$\vp$) in the notation when there is no risk of confusion. For example,
\begin{align}
  \hspace*{-1ex}
  H_1=&
  \begin{bmatrix}
    m_0 & m_1\\
    m_1 & m_2\\
  \end{bmatrix},~&
  B_1=&
  \begin{bmatrix}
    m_1 & m_2\\
    m_2 & m_3\\
  \end{bmatrix},
  \label{eq:H1B1}\\
  \hspace*{-1ex}
  H_2=&
  \begin{bmatrix}
    m_0 & m_1 & m_2\\
    m_1 & m_2 & m_3\\
    m_2 & m_3 & m_4\\
  \end{bmatrix},\!\!&
  B_2=&
  \begin{bmatrix}
    m_1 & m_2 & m_3\\
    m_2 & m_3 & m_4\\
    m_3 & m_4 & m_5\\
  \end{bmatrix}.
  \label{eq:H2B2}
\end{align}
From the definition of the Hankel matrices, one can prove the following result 
on the relations between $\cM_n,\cM_n^+$ and $H_k,B_k$; see 
Appendix~\ref{app:moment} for details.

\begin{lemma}
  (a) A necessary condition for $\vm^{(n)}=(m_0,m_1,\dots,m_n)\in\cM_n$ is that 
  $H_{\lfloor\frac{n}{2}\rfloor}\ge 0$.\\
  (b) A necessary condition for $\vm^{(n)}=(m_0,m_1,\dots,m_n)\in\cM_n^+$ is 
  that $H_{\lfloor\frac{n}{2}\rfloor}\ge 0$ and 
  $B_{\lfloor\frac{n-1}{2}\rfloor}\ge 0$.\\
  \label{lem:moment}
\end{lemma}

By applying Lemma~\ref{lem:moment} to the PT-moment problem, we obtain a family 
of criteria for entanglement detection.

\begin{theorem}
  Let $p_k=\Tr[(\rho^{T_A}_{AB})^k]$ for $k=1,2,\dots,n$, then a necessary 
  condition for $\rho_{AB}$ being a separable state is that 
  $B_{\lfloor\frac{n-1}{2}\rfloor}(\vp)\ge 0$.
  \label{thm:pnPPT}
\end{theorem}

Before preceding, we have a few remarks on Lemma~\ref{lem:moment} and 
Theorem~\ref{thm:pnPPT}. First, the conditions are almost sufficient in 
Lemma~\ref{lem:moment}. If we consider the moments $\vm^{(n)}$ in the closure 
of $\cM_n$ or $\cM_n^+$, then the conditions of the positivity of Hankel 
matrices are also sufficient in Lemma~\ref{lem:moment}; see 
Appendix~\ref{app:moment} for details. Because of the finite precision in 
actual experiments, this also means that Theorem~\ref{thm:pnPPT} is the best 
criterion when relaxing the PT-moment problem to the classical moment problems.

Second, although the condition $H_{\lfloor\frac{n}{2}\rfloor}\ge 0$ is also 
necessary for $\rho_{AB}$ being separable, it does not give an entanglement 
criterion as this condition is satisfied by any (separable or entangled) state 
according to Lemma~\ref{lem:moment}(a).

Third, by noting that $p_1=1$, the lowest-order criterion from 
Theorem~\ref{thm:pnPPT}, $B_1\ge 0$, gives that $p_3\ge p_2^2$, which is 
exactly the $p_3$-PPT condition in Eq.~\eqref{eq:p3PPT} from 
Ref.~\cite{Elben.etal2020}.  When $k>1$, $B_k$ gives stronger criteria for 
entanglement detection.  Accordingly, we call the condition
\begin{equation}
  \rho_{AB}\in\Sep~\Rightarrow~
  B_{\lfloor\frac{n-1}{2}\rfloor}(\vp)\ge 0,
  \label{eq:pnPPT}
\end{equation}
$p_n$-PPT criterion for $n=3,5,7,\dots$. The power of the $p_n$-PPT criteria 
will be illustrated with examples after we describe the optimal method for the 
PT-moment problem.

Last, we would like to point out that although higher-order criteria 
$p_n^{n-2}\ge p_{n-1}^{n-1}$ were also proposed in Ref.~\cite{Elben.etal2020}, 
they usually cannot detect more entangled states than the $p_3$-PPT criterion.  
In Appendix~\ref{app:comparision}, we show that these inequalities are strictly 
weaker than the $p_n$-PPT criteria from Theorem~\ref{thm:pnPPT} and explain why 
these inequalities are usually much weaker.

%%%%%%%%%%%%%%%%%%%%%%%%%%%%%%%%%%%%%%%%%%%%%%%%%%%%%%%%%%%%%%%%%%%%%%%%%%%%%%
% \section{Optimal solution to the PT-moment problem}
\textit{Optimal solution to the PT-moment problem.---}%
%%%%%%%%%%%%%%%%%%%%%%%%%%%%%%%%%%%%%%%%%%%%%%%%%%%%%%%%%%%%%%%%%%%%%%%%%%%%%%
%
Theorem~\ref{thm:pnPPT} already provides a family of strong entanglement 
criteria, but they are not optimal. This is because in 
Eqs.~(\ref{eq:momentSpace},\,\ref{eq:momentSpacePSD}) $\sigma$ can be 
arbitrary, but in the PT-moment problem $\sigma$ is always $\I$. In the 
following, we give an optimal solution to the PT-moment problem.

By writing the spectrum of $\rho_{AB}^{T_A}$ as $(x_1,x_2,\dots,x_d)$, one can 
easily see that the PT-moment problem is equivalent to characterizing the set
\begin{equation}
  \cT_n^+
  ={\Big\{\vp^{(n)}
  \bmid\sum_{i=1}^dx_i^k=p_k,~x_i\ge 0\Big\}}.
  \label{eq:cTp}
\end{equation}
Indeed, for any $\vp^{(n)}\in\cT_n^+$ a compatible separable state can be 
constructed as follows: Relabel $x_i$ for $i=1,2,\dots,d$ as $x_{\alpha\beta}$ 
for $\alpha=1,2,\dots,d_A$ and $\beta=1,2,\dots,d_B$;
then construct a separable state 
$\rho_{AB}=\sum_{\alpha,\beta}x_{\alpha\beta}\ket{\alpha}\bra{\alpha}\otimes 
\ket{\beta}\bra{\beta}$, where $\ket{\alpha}, \ket{\beta}$ are states in the 
computational basis. This state has $p_k=\Tr[(\rho_{AB}^{T_A})^k]$ for 
$k=0,1,\dots,n$. For convenience, we also define the more general set
\begin{equation}
  \cT_n
  ={\Big\{\vp^{(n)}
  \bmid\sum_{i=1}^dx_i^k=p_k,~x_i\in\dR\Big\}}.
  \label{eq:cT}
\end{equation}
Hereafter, the eigenvalues $(x_1,x_2,\dots,x_d)$ are always assumed to be 
sorted in descending order, unless otherwise stated. In 
Eqs.~(\ref{eq:cTp},\,\ref{eq:cT}), the dimension $d=\dim(\cH_A\otimes\cH_B)$ is 
considered as fixed, but actually the optimal entanglement criteria in the 
following, e.g., Eq.~\eqref{eq:p3opt}, do not depend on $d$ anymore.

The key idea of the optimal criteria is to consider the following optimization,
\begin{equation}
  \begin{aligned}
    &\minover[x_i]/\maxover[x_i] && \hat{p}_n:=\sum_{i=1}^d x_i^n\\
    &\subto && \sum_{i=1}^d x_i^k=p_k ~~\text{for}~k=1,2,\dots,n-1,\\
    &       && x_i\ge 0 ~~\text{for}~i=1,2,\dots,d.
  \end{aligned}
  \label{eq:polyopt}
\end{equation}
Note that this may also be viewed as a minimization or maximization of the 
R\'enyi or Tsallis entropy of order $n$ under the constraint that the entropies 
for lower integer orders are fixed. Suppose that the solutions are given by 
$\hat{p}_n^{\min}$ and $\hat{p}_n^{\max}$, respectively, then 
$p_n\in[\hat{p}_n^{\min},\hat{p}_n^{\max}]$ provides a necessary condition for 
$\rho_{AB}$ being separable. If one can further show that all 
$p_n\in[\hat{p}_n^{\min},\hat{p}_n^{\max}]$ are attainable by some 
$(x_1,x_2,\dots,x_d)$ from a separable state, this will imply the sufficiency 
of the condition. As Eq.~\eqref{eq:polyopt} is a polynomial optimization, the 
sum-of-squares hierarchy can, in principle, be used for approximating the 
bounds \cite{Lasserre2001,Parrilo2000}.  Remarkably, an alternative 
sum-of-squares method was used in Ref.~\cite{DelasCuevas.etal2020} for bounding 
the negative eigenvalues from moments. Here, instead of using these 
approximation methods, we propose an exact method for solving 
Eq.~\eqref{eq:polyopt} analytically.

We start from the simplest case $n=3$. As shown in Appendix~\ref{app:optimal}, 
the maximum and minimization are achieved by
\begin{align}
  \vx_3^{\max}&=(x_1,x_2,x_2,\dots,x_2),
  \label{eq:p3max}\\
  \vx_3^{\min}&=(x_1,x_1,\cdots,x_1,x_{\alpha+1},0,0,\dots,0),
  \label{eq:p3min}
\end{align}
respectively, where $x_1$ appears $\alpha=\lfloor 1/p_2\rfloor$ times in 
Eq.~\eqref{eq:p3min}.  Thus, we obtain the following necessary and sufficient 
condition for the PT-moment problem of order three.

\begin{theorem}
  (a) There exists a $d$-dimensional separable state $\rho_{AB}$ satisfying 
  that $p_k=\Tr[(\rho^{T_A}_{AB})^k]$ for $k=1,2,3$, if and only if
  \begin{align}
    \label{eq:p3opt1}
    &p_1=1,\quad\frac{1}{d}\le p_2\le 1,\\
    \label{eq:p3opt2}
    &p_3\le [1-(d-1)y]^3+(d-1)y^3,\\
    &p_3\ge\alpha x^3+(1-\alpha x)^3,
  \end{align}
  where  $\alpha=\lfloor\frac{1}{p_2}\rfloor$, 
  $x=\frac{\alpha+\sqrt{\alpha[p_2(\alpha+1)-1]}}{\alpha(\alpha+1)}$, and 
  $y=\frac{d-1-\sqrt{(d-1)(p_2d-1)}}{d(d-1)}$.\\
  (b) More importantly, suppose that the $p_k$ for $k=1,2,3$ are PT-moments 
  from a quantum state. Then, they are compatible with a separable state if and 
  only if
  \begin{equation}
    p_3\ge \alpha x^3+(1-\alpha x)^3,
    \label{eq:p3opt}
  \end{equation}
  where $\alpha$ and $x$ are as above.
  \label{thm:p3opt}
\end{theorem}

Mathematically speaking, Theorem~\ref{thm:p3opt}(a) fully characterizes the set 
$\cT_3^+$, while Theorem~\ref{thm:p3opt}(b) characterizes the difference 
between $\cT_3^+$ and $\cT_3\setminus\cT_3^+$. In other words, 
Eqs.~(\ref{eq:p3opt1},\,\ref{eq:p3opt2}) are satisfied by any (separable or 
entangled) state. In practice, $p_k$ are usually obtained from experiments, 
hence Eq.~\eqref{eq:p3opt} should be used for entanglement detection. Thus, we 
will refer to Eq.~\eqref{eq:p3opt} as the $p_3$-OPPT (optimal PPT) criterion.  
Again, we emphasize that the $p_3$-OPPT criterion is dimension-independent.

According to Eq.~\eqref{eq:polyopt}, this method is not restricted to the case 
$n=3$. For example, when $n=4$, the maximum and minimum are achieved by
\begin{align}
  \vx_4^{\max}=&(x_1,x_2,x_2,\dots,x_2,
  x_{\beta+2},0,0,\dots,0),
  \label{eq:p4max}\\
  \vx_4^{\min}=&(x_1,x_1\cdots,x_1,x_{\gamma+1},
    x_{\gamma+2},x_{\gamma+2},\dots,x_{\gamma+2}),
  \label{eq:p4min}
\end{align}
respectively, where $\beta$ and $\gamma$ are some fixed integers. However, an 
important difference to the case $n=3$ is that although solving the problem 
analytically is still possible, writing down the optimal values is no longer 
straightforward.  This is because the roots of higher-order polynomials are 
much more complicated \cite{Abel1824}. In Appendix~\ref{app:optimal}, we 
describe the general procedure for solving the optimization problems in 
Eq.~\eqref{eq:polyopt}. We also provide the computer code for $n=3,4,5$ 
\cite{SuppArxiv}.

%%%%%%%%%%%%%%%%%%%%%%%%%%%%%%%%%%%%%%%%%%%%%%%%%%%%%%%%%%%%%%%%%%%%%%%%%%%%%
% \section{Examples}
\textit{Examples.---}%
%%%%%%%%%%%%%%%%%%%%%%%%%%%%%%%%%%%%%%%%%%%%%%%%%%%%%%%%%%%%%%%%%%%%%%%%%%%%%
%
Before discussing the examples, we show how to quantify the violation of the 
$p_n$-PPT and $p_n$-OPPT criteria. Analogous to the PPT criterion, we use the 
negativity \cite{Zyczkowski.etal1998,Vidal.Werner2002} to quantify the 
violation of $p_n$-PPT criteria. For $n=3,5,7,\dots$, we define
\begin{equation}
  \mathcal{N}_n(\rho_{AB})=\frac{1}{2}\big\lVert 
  B_{\lfloor\frac{n-1}{2}\rfloor}(\vp)\big\rVert
  -\frac{1}{2}\Tr\big[B_{\lfloor\frac{n-1}{2}\rfloor}(\vp)\big],
  \label{eq:pnViolation}
\end{equation}
i.e., the absolute sum of the negative eigenvalues of 
$B_{\lfloor\frac{n-1}{2}\rfloor}$, where $\norm{\cdot}$ denotes the trace norm.  
For the $p_n$-OPPT, we quantify the violation via
\begin{equation}
  \mathcal{O}_n(\rho_{AB})=\max\left\{p_n^{\min}-p_n,~p_n-p_n^{\max},~0\right\},
  \label{eq:optViolation}
\end{equation}
for $n=3,4,5,\dots$. Remarkably, although both the $p_n$-PPT and $p_n$-OPPT 
criteria can be viewed as hierarchical entanglement criteria based on 
PT-moments, there are two important distinctions. First, the $p_n$-PPT criteria 
only work when $n$ is odd, while the $p_n$-OPPT criteria work whenever $n\ge 
3$. Second, $\mathcal{N}_n(\rho_{AB})$ in Eq.~\eqref{eq:pnViolation} is 
well-defined for any $\rho_{AB}$, while $\mathcal{O}_n(\rho_{AB})$ only exists 
when $\mathcal{O}_{n-1}(\rho_{AB})=0$, i.e., the optimization problems in 
Eq.~\eqref{eq:polyopt} are feasible.

%%%%%%%%%%%%%%%%%%%%%%%%%%%%%%%%%%%%%%%%%%%%%%%%%%%%%%%%%%%%%%%%%%%%%%%%%%%%%%%
\begin{table}
  \centering
  \begin{tabular}{|c||c|c|c|c|c|c|c|}
    \hline
    $D$ & NPT &NPT$3$ & ONPT$3$ & ONPT$4$ & NPT$5$ & ONPT$5$\\
    \hline% \hline
    2 &75.68\% & 25.53\% & 39.97\% & 75.68\% & 64.78\% & 75.68\%\\
    % \hline
    3 &99.99\% & 25.32\% & 39.46\% & 91.63\% & 97.51\% & 98.97\%\\
    % \hline
    4 &100\% & 23.29\% & 33.69\% & 98.68\% & 100.00\% & 100.00\%\\
    % \hline
    5 &100\% & 21.80\% & 34.54\% & 99.95\% & 100\% & 100\%\\
    % \hline
    6 &100\% & 20.93\% & 31.20\% & 100.00\% & 100\% & 100\%\\
    \hline% \hline
  \end{tabular}
  \caption{%
    Fraction of (small) $D\times D$ states in the Hilbert-Schmidt distribution 
    (1,000,000 samples) that can be detected with various criteria. Here, NPT 
    denotes the states violating the PPT criterion, NPT$n$ (NPT$3$, NPT$5$) 
    denotes the states violating the $p_n$-PPT criterion in 
    Eq.~\eqref{eq:pnPPT}, and ONPT$n$ (ONPT$3$, ONPT$4$, ONPT$5$) denotes the 
    states violating the $p_n$-OPPT criterion.
  }
  \label{tab:HS}
\end{table}
%%%%%%%%%%%%%%%%%%%%%%%%%%%%%%%%%%%%%%%%%%%%%%%%%%%%%%%%%%%%%%%%%%%%%%%%%%%%%%%
%
To show the power of our criteria, we first investigate the entanglement of 
randomly generated states. Here, we sample the random $D\times D$ states 
($\dim(\cH_A)=\dim(\cH_B)=D$) with the Hilbert-Schmidt distribution 
\cite{Zyczkowski.etal2011}. In Table~\ref{tab:HS}, we show the results when $D$ 
is small ($D=2,3,4,5,6$); additional results when $D$ is large 
($D=10,20,30,40$) are shown in Appendix~\ref{app:numerical}.

From the sampling, one can see a few remarkable advantages of our criteria.  
First, most of the entangled states can already be detected by the $p_5$-PPT or 
the $p_4$-OPPT criterion. Second, although the $p_3$-PPT and $p_3$-OPPT 
criteria are both based on the PT-moments $p_2$ and $p_3$, the optimal 
criterion $p_3$-OPPT is significantly stronger than the $p_3$-PPT criterion in 
Ref.~\cite{Elben.etal2020}. Furthermore, the optimal criterion not only detects 
more entangled states, but also the violation is more significant as shown in 
Appendix~\ref{app:numerical}. Third, compared with the usual entanglement 
witness method, our criteria have the advantage that neither common reference 
frames nor prior information is needed for the entanglement detection 
\cite{Ketterer.etal2019,Elben.etal2020}.  Also, compared with the widely-used 
fidelity-based entanglement witness, many more entangled states can be detected 
by comparing Table~\ref{tab:HS} with the results in 
Refs.~\cite{Weilenmann.etal2020,Guehne.etal2021}.

%%%%%%%%%%%%%%%%%%%%%%%%%%%%%%%%%%%%%%%%%%%%%%%%%%%%%%%%%%%%%%%%%%%%%%%%%%%%%%%
\begin{figure}
  \centering
  \includegraphics[width=.45\textwidth]{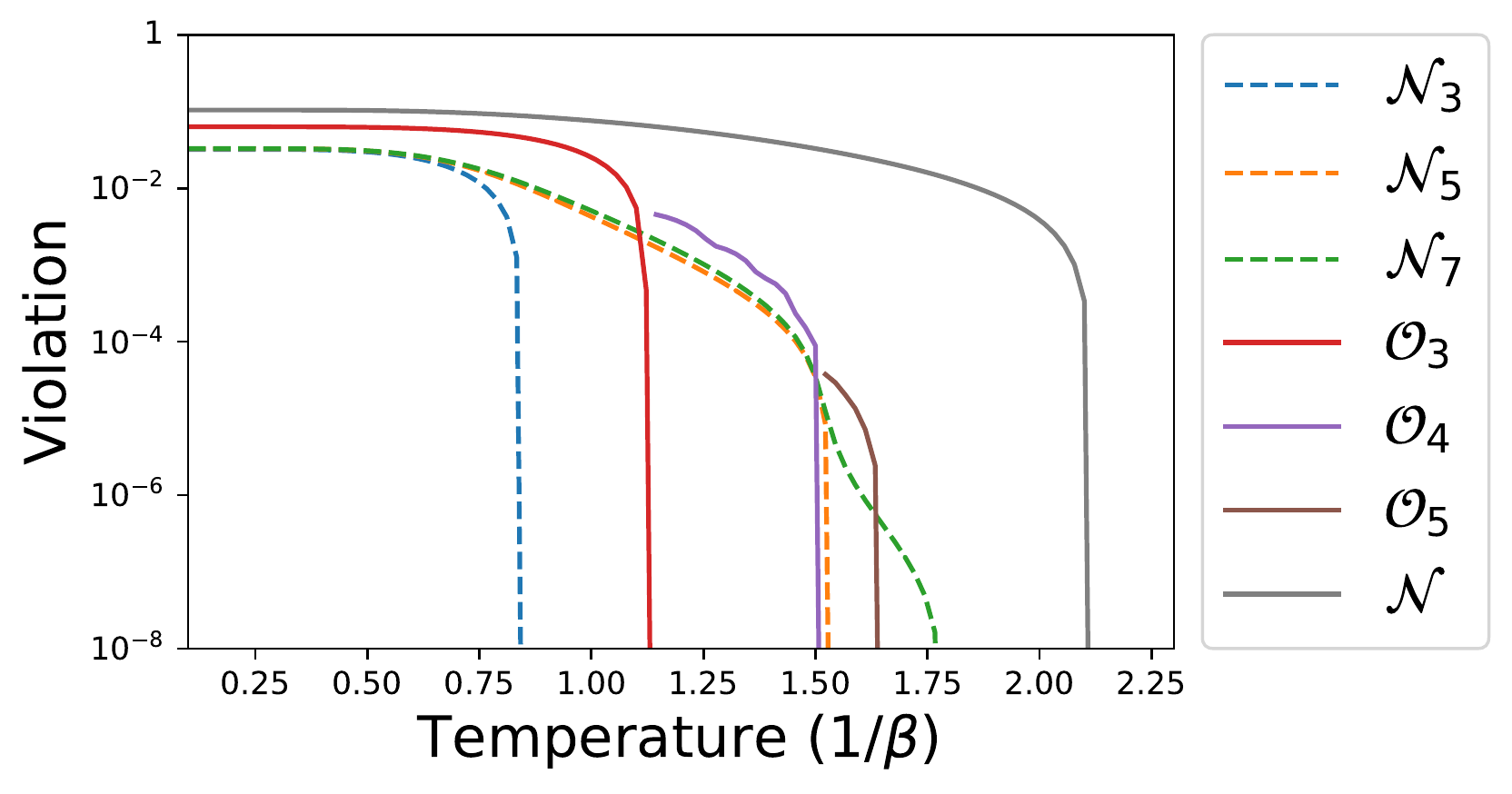}
  \caption{%
    The strength of different PT-moment-based entanglement criteria.  Here, we 
    choose the parameters $J=1$ and $g=2.5$ for a $10$-qubit system.  The 
    entanglement for the bipartition $(1,2,\dots,5|6,7,\dots,10)$ is 
    considered. The violations $\cN_n$ and $\cO_n$ are defined in 
    Eqs.~(\ref{eq:pnViolation},\,\ref{eq:optViolation}), and $\cN$ is the 
    negativity of entanglement \cite{Vidal.Werner2002}.
  }
  \label{fig:TansIsing}
\end{figure}
%%%%%%%%%%%%%%%%%%%%%%%%%%%%%%%%%%%%%%%%%%%%%%%%%%%%%%%%%%%%%%%%%%%%%%%%%%%%%%%

For the second example, we consider the one-dimensional quantum Ising model in 
a transverse magnetic field,
\begin{equation}
  H=-J\Big(\sum_{i=1}^N\sigma_i^z\sigma_{i+1}^z
  +g\sum_{i=1}^N\sigma_i^x\Big),
  \label{eq:TransIsing}
\end{equation}
with the periodic boundary condition ($\sigma^z_{N+1}=\sigma_1^z$), where $J$ 
corresponds to the coupling strength and $g$ is the relative strength of the 
external magnetic field. We study the entanglement of the thermal equilibrium 
(Gibbs) state $\rho(\beta)=\me^{-\beta H}/Z(\beta)$,
where $Z(\beta)=\Tr[\me^{-\beta H}]$ is the partition function and
$\beta$ is the inverse temperature. The strength of different PT-moment-based 
entanglement criteria for this model is illustrated in 
Fig.~\ref{fig:TansIsing}.

At last, we would like to note that the example in Fig.~\ref{fig:TansIsing} 
also illustrates an important challenge for testing the PT-moment-based 
criteria. That is, the violations can become very small for higher-order 
criteria. Indeed, this is not specific to the PT-moments but also the other 
moment-based methods.  The fundamental reason is that the \mbox{(PT-)moments} 
decrease exponentially as $n$ goes large.  This can be easily seen from the 
relation that $\abs{\Tr(X^n)}\le[\Tr(X^2)]^{\frac{n}{2}}$
for any Hermitian operator $X$ and $n\ge 2$
\footnote{This follows directly from the Schur convexity 
\cite{Bhatia1997} of $f(\vy)=\sum_iy_i^{n/2}$ for $y_i\in[0,+\infty)$.}.
In the PT-moment problem, $\Tr[(\rho_{AB}^{T_A})^2]=\Tr[\rho_{AB}^2]$, which is 
the purity of the state. Hence, the violations in Fig.~\ref{fig:TansIsing} 
become small (compared to the PPT criterion) if the temperature increases.  
Still, it should be remembered that a small violation is, in general, not 
connected to a statistically insignificant violation 
\cite{vanDam.etal2005,Acin.etal2005,Jungnitsch.etal2010};
see Appendix~\ref{app:statistics} for more discussions.
This difference also means that the numerical values of the violations in 
Eqs.~(\ref{eq:pnViolation},\,\ref{eq:optViolation}) should not be directly 
compared with each other.

%%%%%%%%%%%%%%%%%%%%%%%%%%%%%%%%%%%%%%%%%%%%%%%%%%%%%%%%%%%%%%%%%%%%%%%%
% \section{Conclusion}
\textit{Conclusion.---}%
%%%%%%%%%%%%%%%%%%%%%%%%%%%%%%%%%%%%%%%%%%%%%%%%%%%%%%%%%%%%%%%%%%%%%%%%
%
We have developed two systematic methods for detecting entanglement from 
PT-moments. The first method is based on the classical moment problems, whose 
lowest order gives the $p_3$-PPT criterion in Ref.~\cite{Elben.etal2020} and 
higher orders provide strictly stronger criteria. The second method is the 
optimal method, which gives necessary and sufficient conditions for 
entanglement detection based on PT-moments. We demonstrated that our criteria 
are significantly better than existing criteria for physically relevant states.

For the future research, there are several possible directions.  First, one may 
extend the presented theory by taking, instead of the transposition, other 
positive but not completely positive maps. This may allow one to characterize 
entanglement in quantum states that escape the detection by the PPT criterion.  
Second, for the analysis of current experiments, it would be highly  desirable 
to extend the presented theory to the characterization of multiparticle 
entanglement. Indeed, potential generalizations of the PPT criterion for the 
multiparticle case exist \cite{Jungnitsch.etal2011}, but how to evaluate this 
using randomized measurements remains an open question.

% \begin{acknowledgments}
  We would like to thank Andreas Ketterer, H. Chau Nguyen, Timo Simnacher, and  
  Nikolai Wyderka for discussions. This work was supported by the Deutsche 
  Forschungsgemeinschaft (DFG, German Research Foundation, project numbers 
  447948357 and 440958198), the Sino-German Center for Research Promotion, and 
  the ERC (Consolidator Grant 683107/TempoQ).

  \textit{Note added.}---While finishing this manuscript, we became aware of 
  a related work by A. Neven \textit{et al.} \cite{Neven.etal2021}.
% \end{acknowledgments}

%%%%%%%%%%%%%%%%%%%%%%%%%%%%%%%%%%%%%%%%%%%%%%%%%%%%%%%%%%%%%%%%%%%%%%%%%%%%%
% Appendices
\onecolumngrid
\vspace{2em}
% \clearpage
\appendix
\newtheorem{manualtheoreminner}{Theorem}
\newenvironment{manualtheorem}[1]{%
  \renewcommand\themanualtheoreminner{#1}%
\manualtheoreminner}{\endmanualtheoreminner}
\newtheorem{manuallemmainner}{Lemma}
\newenvironment{manuallemma}[1]{%
  \renewcommand\themanuallemmainner{#1}%
\manuallemmainner}{\endmanuallemmainner}
%%%%%%%%%%%%%%%%%%%%%%%%%%%%%%%%%%%%%%%%%%%%%%%%%%%%%%%%%%%%%%%%%%%%%%%%%%%%%

%%%%%%%%%%%%%%%%%%%%%%%%%%%%%%%%%%%%%%%%%%%%%%%%%%%%%%%%%%%%%%%%%%%%%%%%%%%%%
\section{The moment problems}\label{app:moment}
%%%%%%%%%%%%%%%%%%%%%%%%%%%%%%%%%%%%%%%%%%%%%%%%%%%%%%%%%%%%%%%%%%%%%%%%%%%%%

In this appendix, we show that the Hankel matrices give almost necessary and 
sufficient conditions for the moment problems. These results follow from 
well-known results in the classical moment problems, which are expressed in the 
language of measure theory; see, for example, 
Refs.~\cite[Chapter~3]{Lasserre2010} and \cite[Chapter~9]{Schmuedgen2017}.  
Here, we give an elementary proof from the point of view of quantum theory. In 
addition, for the sufficiency part (Lemma~\ref{lem:Hankel} in the following), 
we consider the closures $\Cl(\cM_n)$ and $\Cl(\cM_n^+)$ in order to avoid the 
complicated rank and range conditions. We start from the proof of 
Lemma~\ref{lem:moment}.

\begin{manuallemma}{\ref{lem:moment}}
  (a) A necessary condition for $\vm^{(n)}=(m_0,m_1,\dots,m_n)\in\cM_n$ is that 
  $H_{\lfloor\frac{n}{2}\rfloor}\ge 0$.\\
  (b) A necessary condition for $\vm^{(n)}=(m_0,m_1,\dots,m_n)\in\cM_n^+$ is 
  that $H_{\lfloor\frac{n}{2}\rfloor}\ge 0$ and 
  $B_{\lfloor\frac{n-1}{2}\rfloor}\ge 0$.\\
  % \label{lem:moment}
\end{manuallemma}

\begin{proof}
  We take advantage of the Hilbert-Schmidt inner product in the operator space
  \begin{equation}
    \langle X, Y\rangle := \Tr(X^\dagger Y).
    \label{eq:HilbertSchmidt}
  \end{equation}
  Now, consider the sequence of operators 
  $\vv=(\rho^{\frac{1}{2}},\rho^{\frac{1}{2}}X,\dots,
  \rho^{\frac{1}{2}}X^{\lfloor\frac{n}{2}\rfloor})$, and similarly the  
  sequence of operators
  $\vu=(\rho^{\frac{1}{2}}X^{\frac{1}{2}},\rho^{\frac{1}{2}}X^{\frac{3}{2}},
  \dots,\rho^{\frac{1}{2}}X^{\lfloor\frac{n-1}{2}\rfloor+\frac{1}{2}})$ when 
  $X\ge 0$. Then, the Gram matrices for $\vv$ and $\vu$ are given by
  \begin{align}
    \langle v_i, v_j\rangle
    &=\Tr(X^i\rho^{\frac{1}{2}}\rho^{\frac{1}{2}}X^j)
    =\Tr(\rho X^{i+j})
    =m_{i+j},
    \label{eq:Gram}\\
    \langle u_i, u_j\rangle
    &=\Tr(X^{i+\frac{1}{2}}\rho^{\frac{1}{2}}
    \rho^{\frac{1}{2}}X^{j+\frac{1}{2}})
    =\Tr(\rho X^{i+j+1})
    =m_{i+j+1},
    \label{eq:GramPSD}
  \end{align}
  which are just the Hankel matrices $H_{\lfloor\frac{n}{2}\rfloor}$ and 
  $B_{\lfloor\frac{n-1}{2}\rfloor}$. As Gram matrices are always positive 
  semidefinite \cite{Horn.Johnson2012}, we get the results:
  (a) $H_{\lfloor\frac{n}{2}\rfloor}\ge 0$ when $\rho\ge 0$;
  (b) $H_{\lfloor\frac{n}{2}\rfloor}\ge 0$ and 
  $B_{\lfloor\frac{n-1}{2}\rfloor}\ge 0$ when $\rho\ge 0$ and $X\ge 0$.
\end{proof}

\begin{lemma}
  (a) A necessary and sufficient condition for 
  $\vm^{(n)}=(m_0,m_1,\dots,m_n)\in\Cl(\cM_n)$ is that 
  $H_{\lfloor\frac{n}{2}\rfloor}\ge 0$.\\
  (b) A necessary and sufficient condition for 
  $\vm^{(n)}=(m_0,m_1,\dots,m_n)\in\Cl(\cM_n^+)$ is that 
  $H_{\lfloor\frac{n}{2}\rfloor}\ge 0$
  and $B_{\lfloor\frac{n-1}{2}\rfloor}\ge 0$.\\
  \label{lem:Hankel}
\end{lemma}

\begin{proof}
  From Lemma~\ref{lem:moment}, we get that the positive semidefinite property 
  holds when $\vm^{(n)}\in\cM_n$ or $\vm^{(n)}\in\cM_n^+$. Then the necessity 
  parts of both cases~(a) and (b) follow from the fact that the set of positive 
  semidefinite matrices is closed. Thus, we only need to prove the sufficiency 
  parts.

  To prove case~(a), we first show that when $H_{\lfloor\frac{n}{2}\rfloor}>0$, 
  i.e., $H_{\lfloor\frac{n}{2}\rfloor}$ is strictly positive definite, there 
  exists $\ket{\varphi}$ and $X$ of dimension $\lfloor\frac{n}{2}\rfloor+1$ 
  such that
  \begin{equation}
    m_k=\Tr(\ket{\varphi}\bra{\varphi}X^k)=\bra{\varphi}X^k\ket{\varphi}
    \label{eq:puremoments}
  \end{equation}
  for $k=0,1,2,\dots,n$. For simplicity, we use $\ell$ to denote 
  $\lfloor\frac{n}{2}\rfloor$. The basic idea for proving 
  Eq.~\eqref{eq:puremoments} is to construct a so-called flat extension 
  $\vm^{(2\ell+2)}$ of $\vm^{(n)}$, i.e., find $m_{2\ell+2}$ (choose an 
  arbitrary $m_{2\ell+1}$ when $n$ is even) such that
  \begin{equation}
    H_{\ell+1}\ge 0,\quad
    \rank(H_{\ell+1})=\rank(H_\ell).
    \label{eq:flatExt}
  \end{equation}
  Note that $H_\ell>0$ implies that $\rank(H_\ell)=\ell+1$. From the following 
  decomposition
  \begin{equation}
    H_{\ell+1}
    =
    \begin{bmatrix}
      H_\ell & \vect{\mu}_\ell\\
      \vect{\mu}_\ell^T & m_{2\ell+2}
    \end{bmatrix}
    =
    \begin{bmatrix}
      \I & 0\\
      \vect{\mu}_\ell^T H_\ell^{-1} & 1
    \end{bmatrix}
    \begin{bmatrix}
      H_\ell & 0\\
      0 & m_{2\ell+2}-\vect{\mu}_\ell^T H_\ell^{-1}\vect{\mu}_\ell
    \end{bmatrix}
    \begin{bmatrix}
      \I & H_\ell^{-1}\vect{\mu}_\ell\\
      0 & 1
    \end{bmatrix},
    \label{eq:SchurComplement}
  \end{equation}
  where $\vect{\mu}_\ell=[m_{\ell+1},m_{\ell+2},\dots,m_{2\ell+1}]^T$, one can 
  easily see that Eq.~\eqref{eq:flatExt} is satisfied if $m_{2\ell+2}$ is 
  chosen such that the Schur complement $m_{2\ell+2}-\vect{\mu}_\ell^T 
  H_\ell^{-1}\vect{\mu}_\ell$ equals to zero.

  From Eq.~\eqref{eq:flatExt} and $\rank(H_\ell)=\ell+1$, we can construct 
  $\ket{\varphi_i}\in\dR^{\ell+1}$ for $i=0,1,\dots,\ell+1$, such that
  \begin{equation}
    \braket{\varphi_i}{\varphi_j}
    =[H_{\ell+1}]_{ij}
    =m_{i+j},
    \label{eq:defphi}
  \end{equation}
  where $\ket{\varphi_i}$ may not be normalized. Then, the assumption that 
  $H_{\ell}$ is of full rank implies that $\{\ket{\varphi_i}\mid 
  i=0,1,\dots,\ell\}$ is a basis for $\dR^{\ell+1}$, and hence there exists 
  a unique matrix $X\in\dR^{(\ell+1)\times(\ell+1)}$ such that
  \begin{equation}
    X\ket{\varphi_i}=\ket{\varphi_{i+1}},
    \quad\text{for}~i=0,1,\dots,\ell.
    \label{eq:defX}
  \end{equation}
  From Eqs.~(\ref{eq:defphi},\,\ref{eq:defX}), we get that
  \begin{equation}
    \bra{\varphi_i}X\ket{\varphi_j}
    =\bra{\varphi_j}X\ket{\varphi_i}
    =m_{i+j+1}\in\dR
    \quad\text{for}~i,j=0,1,\dots,\ell,
  \end{equation}
  which implies that the real matrix $X$ is symmetric. By letting 
  $\ket{\varphi}=\ket{\varphi_0}$, we get that
  \begin{equation}
    X^i\ket{\varphi}=\ket{\varphi_i}\quad\text{for }
    i=0,1,\dots,\ell+1.
    \label{eq:phiX}
  \end{equation}
  Thus, Eq.~\eqref{eq:puremoments} follows directly from 
  Eqs.~(\ref{eq:defphi},\,\ref{eq:phiX}).

  For the general case that $H_{\lfloor\frac{n}{2}\rfloor}(\vm)\ge 0$, we can 
  construct a sequence of moments $\vm_s^{(n)}$ for $s\in\dN$ such that
  \begin{equation}
    H_{\lfloor\frac{n}{2}\rfloor}(\vm_s)>0,\quad
    \lim_{s\to+\infty}\vm_s^{(n)}=\vm^{(n)}.
    \label{eq:momentsLimit}
  \end{equation}
  For example, we can take
  \begin{equation}
    \vm_s^{(n)}=\left(1-\frac{1}{s+1}\right)\vm^{(n)}+\frac{1}{s+1}\vm_0^{(n)},
  \end{equation}
  where $\vm_0^{(n)}$ is an arbitrary sequence of moments such that 
  $H_{\lfloor\frac{n}{2}\rfloor}(\vm_0)>0$, because
  \begin{equation}
    H_{\lfloor\frac{n}{2}\rfloor}(\vm_s)
    =\left(1-\frac{1}{s+1}\right)H_{\lfloor\frac{n}{2}\rfloor}(\vm)
    +\frac{1}{s+1}H_{\lfloor\frac{n}{2}\rfloor}(\vm_0)>0
  \end{equation}
  for any $s\in\dN$. Then, the necessity follows directly from 
  Eqs.~(\ref{eq:puremoments},\,\ref{eq:momentsLimit}).

  The proof for case~(b) is similar. Again, by assuming that 
  $H_{\lfloor\frac{n}{2}\rfloor}>0$ and $B_{\lfloor\frac{n-1}{2}\rfloor}>0$, 
  a similar argument as in Eq.~\eqref{eq:SchurComplement} implies that we can 
  construct a flat extension $\vm^{(2\lfloor\frac{n}{2}\rfloor+2)}$ of 
  $\vm^{(n)}$ such that
  \begin{equation}
    H_{\lfloor\frac{n}{2}\rfloor+1}\ge 0,\quad
    B_{\lfloor\frac{n}{2}\rfloor}\ge 0,\quad
    \rank\left(H_{\lfloor\frac{n}{2}\rfloor+1}\right)
    =\rank\left(H_{\lfloor\frac{n}{2}\rfloor}\right).
    \label{eq:flatExtS}
  \end{equation}
  Still, we can construct $\ket{\varphi}$ and $X$ as in 
  Eqs.~(\ref{eq:defphi},\,\ref{eq:defX}), then we only need to show that $X\ge 
  0$. This follows from that $\{\ket{\varphi_i}\}_{i=0}^\ell$ is a basis, and
  \begin{equation}
    \bra{\varphi_i}X\ket{\varphi_j}=m_{i+j+1}=\left[B_\ell\right]_{ij},\quad
    B_\ell=B_{\lfloor\frac{n}{2}\rfloor}\ge 0.
    \label{eq:Xpos}
  \end{equation}
  The general case that $H_{\lfloor\frac{n}{2}\rfloor}\ge 0$ and 
  $B_{\lfloor\frac{n-1}{2}\rfloor}\ge0$ can be proved similarly as in 
  Eq.~\eqref{eq:momentsLimit}.
\end{proof}

One may wonder whether Lemma~\ref{lem:Hankel} still holds without taking the 
closure, i.e., whether the sets $\cM_n$ and $\cM_n^+$ are closed. A well-known 
counterexample \cite{Lasserre2010,Schmuedgen2017} for $\cM_n$ is 
$\vm^{(4)}=(1,1,1,1,2)$ with
\begin{equation}
  H_2=
  \begin{bmatrix}
    m_0 & m_1 & m_2\\
    m_1 & m_2 & m_3\\
    m_2 & m_3 & m_4\\
  \end{bmatrix}
  =
  \begin{bmatrix}
    1 & 1 & 1\\
    1 & 1 & 1\\
    1 & 1 & 2\\
  \end{bmatrix}
  \ge 0.
\end{equation}
Suppose that there exists (finite- or infinite-dimensional) $\rho$ and $X$ such 
that $\vm^{(4)}=[\Tr(\rho X^k)]_{k=0}^4=(1,1,1,1,2)$. Then, we have that
\begin{equation}
  \Tr[\rho(X-\I)^2]=\Tr[\rho^{\frac{1}{2}}(X-\I)(X-\I)\rho^{\frac{1}{2}}]=0,
\end{equation}
which implies that $\rho^{\frac{1}{2}}(X-\I)=(X-\I)\rho^{\frac{1}{2}}=0$.  
Thus, $\rho X=\rho$ and hence $\rho X^4=\rho$, which is in contradiction to the 
fact that $\Tr(\rho X^4)=m_4=2\Tr(\rho)=2m_0$. However, $\vm^{(4)}=(1,1,1,1,2)$ 
can be approximated ($\varepsilon\to 0^+$) by
\begin{equation}
  \rho=
  \begin{bmatrix}
    1-\varepsilon & 0\\
    0 & \varepsilon
  \end{bmatrix},
  \quad
  X=
  \begin{bmatrix}
    1 & 0\\
    0 & \varepsilon^{-\frac{1}{4}}
  \end{bmatrix},
\end{equation}
or equivalently, with the pure state 
$\ket{\varphi}=\sqrt{1-\varepsilon}\ket{0}+\sqrt{\varepsilon}\ket{1}$.

For $\cM_n^+$, we can also construct a distinct counterexample 
$\vm^{(4)}=(1,1,2,4,9)$, which satisfies that
\begin{equation}
  H_2=
  \begin{bmatrix}
    m_0 & m_1 & m_2\\
    m_1 & m_2 & m_3\\
    m_2 & m_3 & m_4\\
  \end{bmatrix}
  =
  \begin{bmatrix}
    1 & 1 & 2\\
    1 & 2 & 4\\
    2 & 4 & 9\\
  \end{bmatrix}
  >0,\quad
  B_1=
  \begin{bmatrix}
    m_1 & m_2\\
    m_2 & m_3\\
  \end{bmatrix}
  =
  \begin{bmatrix}
    1 & 2\\
    2 & 4\\
  \end{bmatrix}
  \ge 0.
\end{equation}
It is possible to find $\rho$ and $X$ such that $\Tr(\rho X^i)=m_i$ according 
to Eq.~\eqref{eq:puremoments}. However, it is impossible to make $X\ge 0$, 
because
\begin{equation}
  \Tr[\rho X(X-2\I)^2]=\Tr[\rho^{\frac{1}{2}}(X^{\frac{3}{2}}-2X^{\frac{1}{2}})
  (X^{\frac{3}{2}}-2X^{\frac{1}{2}})\rho^{\frac{1}{2}}]=0,
\end{equation}
which implies that $\rho X^2=2\rho X$ and further $\rho X^4=8\rho X$. This is 
in contradiction to the fact that $\Tr(\rho X^4)=9\Tr(\rho X)$.  However, 
$\vm^{(4)}=(1,1,2,4,9)$ can be approximated ($\varepsilon\to 0^+$) by
\begin{equation}
  \rho=
  \begin{bmatrix}
    \frac{1}{2} & 0 & 0\\
    0 & \frac{1}{2}-\varepsilon & 0\\
    0 & 0 & \varepsilon\\
  \end{bmatrix},
  \quad
  X=
  \begin{bmatrix}
    2 & 0 & 0\\
    0 & 0 & 0\\
    0 & 0 & \varepsilon^{-\frac{1}{4}}\\
  \end{bmatrix},
\end{equation}
or equivalently, with the pure state 
$\ket{\varphi}=\sqrt{1/2}\ket{0}+\sqrt{1/2-\varepsilon}\ket{1}
+\sqrt{\varepsilon}\ket{2}$.

%%%%%%%%%%%%%%%%%%%%%%%%%%%%%%%%%%%%%%%%%%%%%%%%%%%%%%%%%%%%%%%%%%%%%%%%%%%%%
\section{Comparison with the higher-order criteria in 
Ref.~\cite{Elben.etal2020}}\label{app:comparision}
%%%%%%%%%%%%%%%%%%%%%%%%%%%%%%%%%%%%%%%%%%%%%%%%%%%%%%%%%%%%%%%%%%%%%%%%%%%%%

In this appendix, we show that the criteria based on the Hankel matrices are 
strictly stronger than the criteria
\begin{equation}
  p_{n}^{n-2}\ge p_{n-1}^{n-1},
  \quad\text{for}~n=3,4,\dots
  \label{eq:ElbenHigh}
\end{equation}
which were first proposed in Ref.~\cite{Elben.etal2020}. This fact can be 
proved by induction from
\begin{equation}
  \begin{bmatrix}
    p_{n-2} & p_{n-1}\\
    p_{n-1} & p_n
  \end{bmatrix}
  \ge 0,
  \quad\text{for}~n=3,4,\dots
  \label{eq:subHankel}
\end{equation}
which in turn follow from the positive semidefinite property of the Hankel 
matrices $H_k$ and $B_k$. Equation~\eqref{eq:subHankel} implies that
\begin{equation}
  p_n\ge 0,\quad p_np_{n-2}\ge p_{n-1}^2.
  \label{eq:subHigh}
\end{equation}
Note that Eq.~\eqref{eq:subHigh} gives the criterion in 
Eq.~\eqref{eq:ElbenHigh} for $n=3$, i.e., the $p_3$-PPT criterion. Now, assume 
that Eq.~\eqref{eq:ElbenHigh} of order $n-1$ is true, i.e., $p_{n-1}^{n-3}\ge 
p_{n-2}^{n-2}$, then Eq.~\eqref{eq:subHigh} implies that
\begin{equation}
  p_n^{n-2}p_{n-1}^{n-3}\ge p_n^{n-2}p_{n-2}^{n-2}\ge p_{n-1}^{2n-4},
\end{equation}
from which Eq.~\eqref{eq:ElbenHigh} of order $n$ follows.

The proof also explains why the criteria in Eq.~\eqref{eq:ElbenHigh} are 
usually much weaker than the $p_n$-PPT criteria based on the Hankel matrices.  
This is because Eq.~\eqref{eq:subHankel} only contains very limited information 
of the positive semidefinite property of the Hankel matrices. Especially, when 
$n$ is even, Eq.~\eqref{eq:subHankel} holds for all states including the 
nonpositive partial transpose (NPT) ones.  Indeed, the higher-order criteria in 
Eq.~\eqref{eq:ElbenHigh} are so weak that we did not find with random sampling 
a single instance of $\rho_{AB}$ for which they are stronger than the $p_3$-PPT 
criterion, although such states can, in principle, be constructed as follows.

The key idea is to find an NPT state $\rho_{AB}$ such that
\begin{equation}
  p_3\ge p_2^2,\quad 
  \lambda_{\min}\left(\rho_{AB}^{T_A}\right)
  +\lambda_{\max}\left(\rho_{AB}^{T_A}\right)<0,
  \label{eq:counterGoal}
\end{equation}
where $p_n=\Tr[(\rho_{AB}^{T_A})^n]$, and $\lambda_{\min}(\rho_{AB}^{T_A})$ and 
$\lambda_{\max}(\rho_{AB}^{T_A})$ are the smallest and largest eigenvalues of 
$\rho_{AB}^{T_A}$, respectively. Thus, $\rho_{AB}$ does not violate the 
$p_3$-PPT criterion, but it will violate the higher-order criterion in 
Eq.~\eqref{eq:ElbenHigh} when $n$ is a large enough odd number. This is because 
when $k$ is large enough
\begin{equation}
  p_{2k+1}^{2k-1}<0<p_{2k}^{2k}.
\end{equation}
Next, we show that if there exists a state $\rho_{A_1B_1}$ in 
$\cH_{A_1}\otimes\cH_{B_1}$ satisfying that
\begin{equation}
  \lambda_{\min}\left(\rho_{A_1B_1}^{T_{A_1}}\right)
  +\lambda_{\max}\left(\rho_{A_1B_1}^{T_{A_1}}\right)<0,
  \label{eq:bigNeg}
\end{equation}
then a state $\rho_{AB}$ can be constructed by adding some noise such that the 
first part in Eq.~\eqref{eq:counterGoal}, i.e., the $p_3$-PPT criterion, is 
also satisfied. Specifically, let
$\cH_A\otimes\cH_B=\cH_{A_1A_2}\otimes\cH_{B_1B_2}$ and
\begin{equation}
  X_{AB}=\rho_{A_1B_1}\otimes\ket{00}\bra{00}_{A_2B_2}
  +2\lambda\sum_{i=1}^N\ket{00}\bra{00}_{A_1B_1}\otimes\ket{ii}\bra{ii}_{A_2B_2}
  +\lambda\sum_{i=N+1}^{2N}\ket{00}\bra{00}_{A_1B_1}\otimes\ket{ii}\bra{ii}_{A_2B_2},
\end{equation}
where $0<\lambda\le\frac{1}{2}\lambda_{\max}(\rho_{A_1B_1}^{T_{A_1}})$ and $N$ 
is an integer to be determined, then
\begin{equation}
  X_{AB}^{T_A}=\rho_{A_1B_1}^{T_{A_1}}\otimes\ket{00}\bra{00}_{A_2B_2}
  +2\lambda\sum_{i=1}^N\ket{00}\bra{00}_{A_1B_1}\otimes\ket{ii}\bra{ii}_{A_2B_2}
  +\lambda\sum_{i=N+1}^{2N}\ket{00}\bra{00}_{A_1B_1}\otimes\ket{ii}\bra{ii}_{A_2B_2}.
\end{equation}
Let
\begin{equation}
  \rho_{AB}=\frac{X_{AB}}{\Tr(X_{AB})},
\end{equation}
then it is straightforward that 
$\lambda_{\min}(\rho_{AB}^{T_A})+\lambda_{\max}(\rho_{AB}^{T_A})<0$ for any 
$N\in\dN$. Let $\tilde{p}_n=\Tr[(\rho_{A_1B_1}^{T_{A_1}})^n]$, then the 
$p_3$-PPT criterion $p_3\ge p_2^2$ for $\rho_{AB}$ is equivalent to that
\begin{equation}
  \Tr[X_{AB}]\Tr[(X_{AB}^{T_A})^3]
  \ge\left[\Tr[(X_{AB}^{T_A})^2]\right]^2
  ~\Leftrightarrow~
  2\lambda^4N^2+(9\lambda^3-10\tilde{p}_2\lambda^2+3\tilde{p}_3 
  \lambda)N+(\tilde{p}_3-\tilde{p}_2^2)\ge0,
\end{equation}
which can be satisfied by choosing a large enough $N$. Thus, we only need to 
show that there exists $\rho_{A_1B_1}$ satisfying Eq.~\eqref{eq:bigNeg}. An 
explicit example is given by the Werner states~\cite{Werner1989}
\begin{equation}
  \rho_{A_1B_1}=\frac{1}{d_1(d_1-1)}(\I-V_{A_1A_2}),
\end{equation}
where $V_{A_1B_1}$ is the swap operator between $\cH_{A_1}$ and $\cH_{B_1}$, 
and $d_1=\dim(\cH_{A_1})=\dim(\cH_{B_1})$. One can easily verify that
\begin{equation}
  \lambda_{\min}\left(\rho_{A_1B_1}^{T_{A_1}}\right)
  +\lambda_{\max}\left(\rho_{A_1B_1}^{T_{A_1}}\right)
  =-\frac{(d_1-2)}{d_1(d_1-1)},
\end{equation}
which is always negative when $d_1\ge 3$.

%%%%%%%%%%%%%%%%%%%%%%%%%%%%%%%%%%%%%%%%%%%%%%%%%%%%%%%%%%%%%%%%%%%%%%%%%%%%%
\section{The optimal method for the PT-moment problem}\label{app:optimal}
%%%%%%%%%%%%%%%%%%%%%%%%%%%%%%%%%%%%%%%%%%%%%%%%%%%%%%%%%%%%%%%%%%%%%%%%%%%%%

In this appendix, we describe in detail the optimal method for the PT-moment 
problem. As explained in the main text, this is equivalent to characterizing 
the difference between
\begin{align}
  \cT_n^+&
  ={\Big\{\vp^{(n)}=(p_0,p_1,\dots,p_n)
  \,\Big|\,\sum_{i=1}^dx_i^k=p_k,~x_i\ge 0\Big\}},\\
  \cT_n&
  ={\Big\{\vp^{(n)}=(p_0,p_1,\dots,p_n)
  \,\Big|\,\sum_{i=1}^dx_i^k=p_k,~x_i\in\dR\Big\}}.
\end{align}
To this end, we consider the following optimization problems,
\begin{equation}
  \begin{aligned}
    &\minover[x_i]/\maxover[x_i] \quad && \hat{p}_n:=\sum_{i=1}^d x_i^n\\
    &\subto && \sum_{i=1}^d x_i^k=p_k \quad ~\text{for}~ k=1,2,\dots,n-1,\\
    &       && x_1\ge x_2\ge\dots\ge x_d\ge 0.
  \end{aligned}
  \label{eq:polyoptA}
\end{equation}
Then, the moment $p_n$ should be within $[\hat{p}_n^{\min},\hat{p}_n^{\max}]$, 
which gives a necessary condition for the moments to originate from a separable 
state. Actually, as shown in the following, this condition is also sufficient 
as there is no local minimum and maximum apart from the global ones. We start 
from the case that  $n=3$ and prove Theorem~\ref{thm:p3opt} from the main text.

\begin{manualtheorem}{\ref{thm:p3opt}}
  (a) There exists a $d$-dimensional separable (PPT) state $\rho_{AB}$ 
  satisfying that $p_k=\Tr[(\rho^{T_A}_{AB})^k]$ for $k=1,2,3$, if and only if
  \begin{equation}
    p_1=1,\quad\frac{1}{d}\le p_2\le 1,\quad
    \alpha x^3+(1-\alpha x)^3\le
    p_3\le [1-(d-1)y]^3+(d-1)y^3,
  \end{equation}
  where  $\alpha=\lfloor\frac{1}{p_2}\rfloor$, 
  $x=\frac{\alpha+\sqrt{\alpha[p_2(\alpha+1)-1]}}{\alpha(\alpha+1)}$, and 
  $y=\frac{d-1-\sqrt{(d-1)(p_2d-1)}}{d(d-1)}$.\\
  (b) More importantly, suppose that the $p_k$ for $k=1,2,3$ are PT-moments 
  from a quantum state. Then, they are compatible with
  a separable (PPT) state if and only if
  \begin{equation}
    p_3\ge \alpha x^3+(1-\alpha x)^3,
  \end{equation}
  where $\alpha$ and $x$ are as above.
\end{manualtheorem}

\begin{proof}
  It is well-known that $p_1=\Tr[\rho_{AB}^{T_A}]=\Tr[\rho_{AB}]=1$, and 
  further the optimization problems in Eq.~\eqref{eq:polyoptA} for $n=3$ are 
  feasible if and only if $1/d\le p_2\le 1$. To solve the optimization problems 
  in Eq.~\eqref{eq:polyoptA} for $n=3$, we start from the simplest nontrivial 
  case that $d=3$. Then, the optimization reads
  \begin{equation}
    \begin{aligned}
      &\minover[x_i]/\maxover[x_i] \quad && \hat{p}_3:=x_1^3+x_2^3+x_3^3\\
      &\subto && x_1+x_2+x_3=p_1,\\
      &       && x_1^2+x_2^2+x_3^2=p_2,\\
      &       && x_1\ge x_2\ge x_3\ge 0,
    \end{aligned}
    \label{eq:polyoptn3d3}
  \end{equation}
  where $p_1$ and $p_2$ are constants and $\hat{p}_3$ is the objective function 
  that we want to optimize. From Eq.~\eqref{eq:polyoptn3d3} we can get how 
  $\hat{p}_3$ varies with $x_i$, i.e., the relations between the differentials 
  $\md\hat{p}_3$ and $\md x_i$,
  \begin{equation}
    \begin{alignedat}{4}
      &\md x_1+\null&&\md x_2+\null&&\md x_3&&=0,\\
      x_1&\md x_1+x_2&&\md x_2+x_3&&\md x_3&&=0,\\
      x_1^2&\md x_1+x_2^2&&\md x_2+x_3^2&&\md x_3&&
      =\frac{1}{3}\md\hat{p}_3.\\
    \end{alignedat}
    \label{eq:polyoptn3d3diff}
  \end{equation}
  This can be viewed as a system of linear equations on $\md x_i$ and can be 
  directly solved by taking advantage of Cramer's rule and the Vandermonde 
  determinant \cite{Horn.Johnson2012}, whenever $x_i$ are all different. This 
  gives the following relations between $\md\hat{p}_3$ and $\md x_i$
  \begin{equation}
    \begin{aligned}
      \md\hat{p}_3=&3(x_1-x_2)(x_1-x_3)\md x_1\\
      =&3(x_2-x_3)(x_2-x_1)\md x_2\\
      =&3(x_3-x_1)(x_3-x_2)\md x_3.\\
    \end{aligned}
    \label{eq:polyoptn3d3dp3}
  \end{equation}
  Recalling that $x_1\ge x_2\ge x_3$ by assumption, 
  Eq.~\eqref{eq:polyoptn3d3dp3} implies that $\md x_i$ are not independent and 
  an alternating relation exists between them. For example, $\md x_1>0$ will 
  result in that $\md x_2<0$ and $\md x_3>0$ (when $x_i$ are all different).  
  Further, Eq.~\eqref{eq:polyoptn3d3dp3} also implies that if the optimization 
  problems in Eq.~\eqref{eq:polyoptn3d3} are feasible, then $\md\hat{p}_3 > 0$ 
  when $x_1$ increases (and thus $x_2$ decreases and $x_3$ increases); and 
  $\md\hat{p}_3<0$ when $x_1$ decreases (and thus $x_2$ increases and $x_3$ 
  decreases). Then, without the boundary condition $x_i\ge 0$, the maximum of 
  $\hat{p}_3$ will be achieved when $x_2=x_3$ and the minimum will be achieved 
  when $x_1=x_2$.  When the boundary condition $x_3\ge 0$ is taken into 
  consideration, the minimum may also be achieved when $x_3$ decreases to zero.

  Note that the above analysis does not depend on the actual values of $p_1$ 
  and $p_2$ (even if $p_1\ne 1$). This implies that if the optimization 
  problems in Eq.~\eqref{eq:polyopt} for
  $\hat{p}_3$ ($n=3$) are feasible, the (local) maximum will be achieved only 
  if
  \begin{equation}
    \vx_3^{\max}=(x_1,\underbrace{x_2,x_2,\dots,x_2}_{d-1~\text{times}}),
    \label{eq:p3maxA}
  \end{equation}
  and the (local) minimum will be achieved only if
  \begin{equation}
    \vx_3^{\min}
    =(\underbrace{x_1,x_1,\cdots,x_1}_{\alpha~\text{times}},
    x_{\alpha+1},\underbrace{0,0,\dots,0}_{d-\alpha-1~\text{times}}).
    \label{eq:p3minA}
  \end{equation}
  This is because for the maximization any tuple $(x_{i_1},x_{i_2},x_{i_3})$ 
  with $i_1<i_2<i_3$ needs to satisfy that $x_{i_2}=x_{i_3}$, and for the 
  minimization it needs to satisfy that $x_{i_1}=x_{i_2}$ or $x_{i_3}=0$.  
  Without loss of generality, we assume that $x_{\alpha+1}\ne x_1$ in 
  Eq.~\eqref{eq:p3minA}, then the integer $\alpha$ is uniquely determined by
  \begin{equation}
    \alpha=\left\lfloor\frac{1}{p_2}\right\rfloor.
  \end{equation}
  This is because the majorization relation
  \begin{equation}
    \Big(\underbrace{\frac{1}{\alpha+1},\frac{1}{\alpha+1},\dots,
      \frac{1}{\alpha+1}}_{\alpha+1~\text{times}},
      \underbrace{
	\vphantom{\frac{1}{\alpha+1}}%vertical placeholder
    0,0,\dots,0}_{d-\alpha-1~\text{times}}\Big)
    \prec(\underbrace{x_1,x_1,\cdots,x_1}_{\alpha~\text{times}},
    x_{\alpha+1},\underbrace{0,0,\dots,0}_{d-\alpha-1~\text{times}})
    \prec\Big(\underbrace{\frac{1}{\alpha},\frac{1}{\alpha},
      \dots,\frac{1}{\alpha}}_{\alpha~\text{times}},
      \underbrace{
	\vphantom{\frac{1}{\alpha}}%vertical placeholder
    0,0,\dots,0}_{d-\alpha~\text{times}}\Big).
    \label{eq:majorization}
  \end{equation}
  and the strict Schur-convexity of $p_2=\sum_{i=1}^dx_i^2$ \cite{Bhatia1997} 
  imply that $1/(\alpha+1)< p_2\le 1/\alpha$.

  Now, from Eq.~\eqref{eq:p3maxA}, we can get that
  \begin{equation}
    \begin{aligned}
      &x_1+(d-1)x_2=1,\\
      &x_1^2+(d-1)x_2^2=p_2,\\
      &x_1\ge x_2\ge 0.
    \end{aligned}
    \label{eq:p3maxls}
  \end{equation}
  Given the feasibility condition $1/d\le p_2\le 1$, Eq.~\eqref{eq:p3maxls} has 
  a unique solution
  \begin{equation}
    x_1=\frac{\sqrt{(d-1)(p_2 d -1)}+1}{d},\quad
    x_2=\frac{d-1-\sqrt{(d-1)(p_2 d-1)}}{d(d-1)}.
    \label{eq:p3maxsol}
  \end{equation}
  From Eq.~\eqref{eq:p3minA}, we can get that
  \begin{equation}
    \begin{aligned}
      &\alpha x_1+x_{\alpha+1}=1,\\
      &\alpha x_1+x_{\alpha+1}^2=p_2,\\
      &x_1\ge x_{\alpha+1}\ge 0,
    \end{aligned}
    \label{eq:p3minls}
  \end{equation}
  which also has a unique solution
  \begin{equation}
    x_1=
    \frac{\alpha +\sqrt{\alpha[p_2(\alpha+1)-1]}}{\alpha(\alpha+1)},\quad
    x_{\alpha+1}=
    \frac{1-\sqrt{\alpha[p_2(\alpha +1)-1]}}{\alpha +1}.
    \label{eq:p3minsol}
  \end{equation}
  So far, we only considered the conditions for local extrema, and found that 
  the minimum and maximum are unique as in 
  Eqs.~(\ref{eq:p3maxsol},\,\ref{eq:p3minsol}). This implies that these are the 
  global extrema. Further, the uniqueness of the extrema and the continuity of  
  $\hat{p}_3$ also imply that the closed feasible region is connected. Thus, 
  all values between the minimum and the maximum are achievable.  All these 
  arguments lead to the optimal result in Theorem~\ref{thm:p3opt}(a).

  For Theorem~\ref{thm:p3opt}(b), as $p_1=\Tr[\rho_{AB}^{T_A}]=\Tr[\rho_{AB}]$ 
  and $p_2=\Tr[(\rho_{AB}^{T_A})^2]=\Tr[\rho_{AB}^2]$, the conditions that 
  $p_1=1$ and $1/d\le p_2\le 1$ are always satisfied by all (separable or 
  entangled) states. To show the redundancy of $p_3\le\hat{p}_3^{\max}$, we 
  need to consider the optimization problems in Eq.~\eqref{eq:polyoptA} without 
  the positivity constraints $x_i\ge 0$. From Eq.~\eqref{eq:polyoptn3d3dp3}, we 
  can see that the maximization is still achieved when $\vx$ is of the form in 
  Eq.~\eqref{eq:p3maxA}. Given the above conditions $p_1=1$ and $1/d\le p_2\le 
  1$, the solution is always positive from Eq.~\eqref{eq:p3maxsol}. Thus, we 
  prove the optimal result in Theorem~\ref{thm:p3opt}(b).
\end{proof}

We note that for the case $n=3$ the bounds in Theorem~\ref{thm:p3opt}(a) can 
also be derived from the optimization of R\'enyi/Tsallis 
entropy~\cite{Berry.Sanders2003}, but our method has the advantages that the 
refined result in Theorem~\ref{thm:p3opt}(b) is given, and more importantly, it 
can be directly generalized to higher-order optimizations in 
Eq.~\eqref{eq:polyoptA}.  We take $n=4$ as an example to illustrate this.  The 
basic idea is still to consider the simplest nontrivial situation that $d=4$ 
first, i.e.,
\begin{equation}
  \begin{aligned}
    &\minover[x_i]/\maxover[x_i] \quad && \hat{p}_4:=x_1^4+x_2^4+x_3^4+x_4^4\\
    &\subto && x_1+x_2+x_3+x_4=p_1,\\
    &       && x_1^2+x_2^2+x_3^2+x_4^2=p_2,\\
    &       && x_1^3+x_2^3+x_3^3+x_4^3=p_3,\\
    &       && x_1\ge x_2\ge x_3\ge x_4\ge 0,
  \end{aligned}
  \label{eq:polyoptn4d4}
\end{equation}
for which an analog of Eq.~\eqref{eq:polyoptn3d3diff} reads
\begin{equation}
  \begin{alignedat}{5}
    &\md x_1+\null&&\md x_2+\null&&\md x_3+\null&&\md x_4&&=0,\\
    x_1&\md x_1+x_2&&\md x_2+x_3&&\md x_3+x_4&&\md x_4&&=0,\\
    x_1^2&\md x_1+x_2^2&&\md x_2+x_3^2&&\md x_3+x_4^2&&\md x_4&&=0,\\
    x_1^3&\md x_1+x_2^3&&\md x_2+x_3^3&&\md x_3
    +x_4^3&&\md x_4&&=\frac{1}{4}\md\hat{p}_4.\\
  \end{alignedat}
  \label{eq:polyoptn4d4diff}
\end{equation}
Then, Cramer's rule and the Vandermonde determinant imply
\begin{equation}
  \begin{aligned}
    \md\hat{p}_4=&4(x_1-x_2)(x_1-x_3)(x_1-x_4)\md x_1\\
    =&4(x_2-x_3)(x_2-x_4)(x_2-x_1)\md x_2\\
    =&4(x_3-x_4)(x_3-x_1)(x_3-x_2)\md x_3\\
    =&4(x_4-x_1)(x_4-x_2)(x_4-x_3)\md x_4.\\
  \end{aligned}
\end{equation}
With a similar argument as for Eq.~\eqref{eq:polyoptn3d3dp3}, we get that when 
the optimization problems in Eq.~\eqref{eq:polyoptn4d4} are feasible, the 
maximum is achieved when $x_2=x_3$ or $x_4=0$, and the minimum is achieved when 
$x_1=x_2$ or $x_3=x_4$. This analysis implies that if the optimization problems 
in Eq.~\eqref{eq:polyopt} for $\hat{p}_4$ are feasible i.e., all the conditions 
in Theorem~\ref{thm:p3opt}(a) are satisfied, the maximum will be achieved only 
if
\begin{equation}
  \vx_4^{\max}=(x_1,\underbrace{x_2,x_2,\dots,x_2}_{\beta~\text{times}},
  x_{\beta+2},\underbrace{0,0,\dots,0}_{d-\beta-2~\text{times}}),
  \label{eq:p4maxA}
\end{equation}
and the minimum will be achieved only if
\begin{equation}
  \vx_4^{\min}
  =(\underbrace{x_1,x_1\cdots,x_1}_{\gamma~\text{times}},x_{\gamma+1},
    \underbrace{x_{\gamma+2},x_{\gamma+2},\dots,x_{\gamma+2}}_{d-\gamma-1
  ~\text{times}}).
  \label{eq:p4minA}
\end{equation}
Furthermore, the vectors $\vx_4^{\max}$ and $\vx_4^{\min}$ are uniquely 
determined.  The argument for the uniqueness is easy but tedious.  Here, we 
only take $\vx_4^{\max}$ as an example to show the basic idea; the uniqueness 
of $\vx_4^{\min}$ can be proved similarly. The task is to prove that the 
equations
\begin{equation}
  \begin{aligned}
    & x_1+\beta x_2+x_{\beta+2}=1,\\
    & x_1^2+\beta x_2^2+x_{\beta+2}^2=p_2,\\
    & x_1^3+\beta x_2^3+ x_{\beta+2}^3=p_3,\\
    & x_1\ge x_2\ge x_{\beta+2}\ge 0.
  \end{aligned}
  \label{eq:maxvar4}
\end{equation}
uniquely determine $\beta,x_1,x_2,x_{\beta+2}$. To this end, we fix $p_2$ but 
treat $p_3$ as a function on $\beta,x_1,x_2,x_{\beta+2}$. We aim to show that 
that $p_3$ is monotonically increasing on $x_1$ (for fixed $p_2$) and thus 
$p_2,p_3$ uniquely determine $x_1$. Then, as in Eq.~\eqref{eq:p3minA}, for 
fixed $p_2$ and $x_1$, the equations
\begin{equation}
  \begin{aligned}
    \beta x_2+x_{\beta+2}=1-x_1,\\
    \beta x_2^2+x_{\beta+2}^2=p_2-x^2_1,\\
  \end{aligned}
\end{equation}
uniquely determine $\beta$, $x_2$, and $x_{\beta+2}$ and hence, given the 
feasibility, $\vx_4^{\max}$ is uniquely determined by $p_2$ and $p_3$.

In the following, we always assume that $p_2$ is fixed. Similarly to 
Eq.~\eqref{eq:polyoptn3d3dp3}, we can show that for fixed $\beta$, $\md p_3>0$ 
if $\md x_1>0$. Then, we only need to show that there is no overlap between the 
ranges of $p_3$ for different $\beta$ (except the extrema). To this end, we 
apply a similar argument as in 
Eqs.~(\ref{eq:polyoptn3d3},\,\ref{eq:polyoptn3d3diff},\,\ref{eq:polyoptn3d3dp3})
to the following optimization for fixed $\beta$,
\begin{equation}
  \begin{aligned}
    &\minover[x_i]/\maxover[x_i] \quad && p_3:=x_1^3+\beta 
    x_2^3+x_{\beta+2}^3\\
    &\subto && x_1+\beta x_2+x_{\beta+2}=1,\\
    &       && x_1^2+\beta x_2^2+x_{\beta+2}^2=p_2,\\
    &       && x_1\ge x_2\ge x_{\beta+2}\ge 0.
  \end{aligned}
\end{equation}
One can show that the minimization is achieved when $x_1=x_2$ or 
$x_{\beta+2}=0$, and the maximization is achieved when $x_2=x_{\beta+2}$. When 
these optimization results are written down consecutively for all possible 
$\beta$, i.e., from $\beta=\alpha-1=\lfloor 1/p_2\rfloor-1$ to $\beta=d-2$, one 
can easily see that $p_3$ is monotonically increasing in following process 
($\nearrow$ denotes increasing and $\searrow$ denotes decreasing)
\begin{equation}
  \begin{aligned}
    &(x_1^\nearrow,\underbrace{x_1^\searrow,
      x_1^\searrow,\dots,x_1^\searrow}_{\alpha-1~\text{times}},
    x_{\alpha+1}^\nearrow,\underbrace{0,0,\dots,0}_{d-\alpha-1~\text{times}})
    &\Rightarrow~~&(x_1^\nearrow,\underbrace{x_2^\searrow,
      x_2^\searrow,\dots,x_2^\searrow}_{\alpha-1~\text{times}},
    x_{\alpha+1}^\nearrow,\underbrace{0,0,\dots,0}_{d-\alpha-1~\text{times}})\\
    \Rightarrow~&(x_1^\nearrow,\underbrace{x_2^\searrow,
      x_2^\searrow,\dots,x_2^\searrow}_{\alpha~\text{times}},
    0^\nearrow,\underbrace{0,0,\dots,0}_{d-\alpha-2~\text{times}})
    &\Rightarrow~~&(x_1^\nearrow,\underbrace{x_2^\searrow,
      x_2^\searrow,\dots,x_2^\searrow}_{\alpha~\text{times}},
    x_{\alpha+2}^\nearrow,\underbrace{0,0,\dots,0}_{d-\alpha-2~\text{times}})\\
    \vdots\quad&\\
    \Rightarrow~&(x_1^\nearrow,\underbrace{x_2^\searrow,
      x_2^\searrow,\dots,x_2^\searrow}_{\beta~\text{times}},
    0^\nearrow,\underbrace{0,0,\dots,0}_{d-\beta-2~\text{times}})
    &\Rightarrow~~&(x_1^\nearrow,\underbrace{x_2^\searrow,
      x_2^\searrow,\dots,x_2^\searrow}_{\beta~\text{times}},
    x_{\beta+2}^\nearrow,\underbrace{0,0,\dots,0}_{d-\beta-2~\text{times}})\\
    \Rightarrow~&(x_1^\nearrow,\underbrace{x_2^\searrow,
      x_2^\searrow,\dots,x_2^\searrow}_{\beta+1~\text{times}},
    0^\nearrow,\underbrace{0,0,\dots,0}_{d-\beta-3~\text{times}})
    &\Rightarrow~~&(x_1^\nearrow,\underbrace{x_2^\searrow,
      x_2^\searrow,\dots,x_2^\searrow}_{\beta+1~\text{times}},
    x_{\beta+3}^\nearrow,\underbrace{0,0,\dots,0}_{d-\beta-3~\text{times}})\\
    \vdots\quad&\\
    \Rightarrow~&(x_1^\nearrow,\underbrace{x_2^\searrow,
      x_2^\searrow,\dots,x_2^\searrow}_{d-2~\text{times}},
    x_d^\nearrow)
    &\Rightarrow~~&(x_1,\underbrace{x_2, x_2,\dots,x_2}_{d-1~\text{times}}).
  \end{aligned}
  \label{eq:p4maxmon}
\end{equation}
Thus, $p_3$ is monotonically increasing on $x_1$ and there is no overlap 
between the ranges of $p_3$ for different $\beta$ (except 
the extrema)

The above arguments also provide a method to determine $\beta$ and then 
$\vx_4^{\max}$ completely. Let $\beta_0$ be the unique solution of the 
following equations
\begin{equation}
  \begin{aligned}
    & x_1+\beta_0x_2=1,\\
    & x_1^2+\beta_0x_2^2=p_2,\\
    & x_1^3+\beta_0x_2^3=p_3,\\
    & x_1\ge x_2\ge 0,~\beta_0\ge 1,
  \end{aligned}
  \label{eq:estbetaA}
\end{equation}
then $\beta=\lfloor\beta_0\rfloor$ and the maximum point $\vx_4^{\max}$ in 
Eq.~\eqref{eq:p4maxA} can be obtained by solving Eq.~\eqref{eq:maxvar4}.

Similarly, let $\gamma_0$ be the unique solution of the following equations 
(unless $p_k=1/d^{k-1}$)
\begin{equation}
  \begin{aligned}
    & \gamma_0x_1+(d-\gamma_0)x_2=1,\\
    & \gamma_0x_1^2+(d-\gamma_0)x_2^2=p_2,\\
    & \gamma_0x_1^3+(d-\gamma_0)x_2^3=p_3,\\
    & x_1\ge x_2\ge 0,~\gamma_0\ge 1.
  \end{aligned}
  \label{eq:estgammaA}
\end{equation}
then $\gamma=\lfloor\gamma_0\rfloor$ and minimum point $\vx_4^{\min}$ in 
Eq.~\eqref{eq:p4minA} can be obtained by solving
\begin{equation}
  \begin{aligned}
    & \gamma x_1+x_{\gamma+1}+(d-\gamma-1)x_{\gamma+2}=1,\\
    & \gamma x_1^2+x_{\gamma+1}^2+(d-\gamma-1)x_{\gamma+2}^2=p_2,\\
    & \gamma x_1^3+x_{\gamma+1}^3+(d-\gamma-1)x_{\gamma+2}^3=p_3,\\
    & x_1\ge x_{\gamma+1}\ge x_{\gamma+2}\ge 0.
  \end{aligned}
  \label{eq:minvar4}
\end{equation}
Thus, given the feasibility, i.e., all the conditions in 
Theorem~\ref{thm:p3opt}(a) are satisfied, the necessary and sufficient 
condition for $\rho_{AB}$ can be separable (PPT) is
\begin{equation}
  \hat{p}_4^{\min}\le p_4\le \hat{p}_4^{\max},
  \label{eq:p4optA}
\end{equation}
where $\hat{p}_4^{\min}$ and $\hat{p}_4^{\max}$ are determined by  
$\vx_4^{\min}$ in  Eq.~\eqref{eq:p4minA} and $\vx_4^{\max}$ in 
Eq.~\eqref{eq:p4maxA}, respectively. Correspondingly, we can show that the 
lower bound is redundant because $\vx_4^{\min}$ is also the minimum without the 
positivity constraints $x_i\ge 0$.
An intuitive way to understand why $\vx_3^{\max}$ and $\vx_4^{\min}$ are 
redundant is that they are extrema that are not on the boundary of nonnegative 
vectors.

With an analogous procedure, we can also solve the optimization problems in 
Eq.~\eqref{eq:polyoptA} for $n=5$, where the maximum will be achieved only if
\begin{equation}
  \vx_5^{\max}=(x_1,
    \underbrace{x_2,x_2,\dots,x_2}_{\kappa~\text{times}},x_{\kappa+2},
  \underbrace{x_{\kappa+3},x_{\kappa+3},\dots,x_{\kappa+3}}_{d-\kappa-2~\text{times}}),
  \label{eq:p5maxA}
\end{equation}
and the minimum will be achieved only if
\begin{equation}
  \vx_5^{\min}=(\underbrace{x_1,x_1,\cdots,x_1}_{\eta~\text{times}},
    x_{\eta+1},\underbrace{x_{\eta+2},
    x_{\eta+2}\dots,x_{\eta+2}}_{\xi~\text{times}},
    x_{\eta+\xi+2},
  \underbrace{0,0,\dots,0}_{d-\eta-\xi-2~\text{times}}).
  \label{eq:p5minA}
\end{equation}
In this case, it is more complicated to determine $\vx^{\max}_5$ and 
$\vx^{\min}_5$. Instead of writing down the complicated formula, we provide the 
Mathematica code for performing the optimizations.

%%%%%%%%%%%%%%%%%%%%%%%%%%%%%%%%%%%%%%%%%%%%%%%%%%%%%%%%%%%%%%%%%%%%%%%%%%%%%
\newpage
\section{Additional numerical results}\label{app:numerical}
%%%%%%%%%%%%%%%%%%%%%%%%%%%%%%%%%%%%%%%%%%%%%%%%%%%%%%%%%%%%%%%%%%%%%%%%%%%%%

%%%%%%%%%%%%%%%%%%%%%%%%%%%%%%%%%%%%%%%%%%%%%%%%%%%%%%%%%%%%%%%%%%%%%%%%%%%%%%%
\begin{figure}[h!]
  \centering
  \includegraphics[width=.4\textwidth]{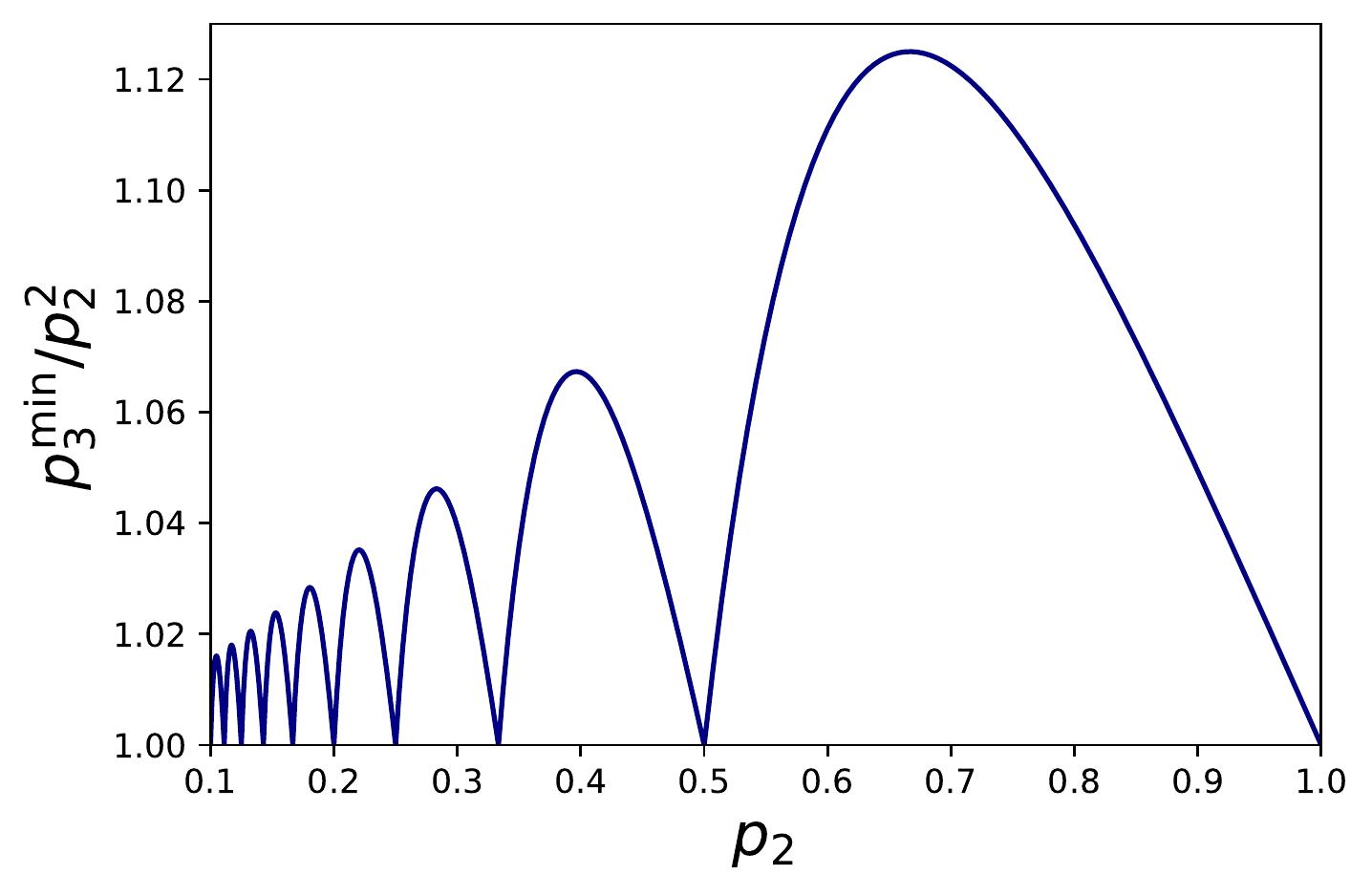}
  \caption{%
    An illustration of the difference between the optimal 
    criterion $p_3$-OPPT in Theorem~\ref{thm:p3opt}(b) and 
    the $p_3$-PPT condition $p_3\ge p_2^2$ in Ref.~\cite{Elben.etal2020}. This 
    shows that the violation of the optimal criterion can be up to $12.5\%$ 
    larger.
  }
  \label{fig:p3ratio}
\end{figure}
%%%%%%%%%%%%%%%%%%%%%%%%%%%%%%%%%%%%%%%%%%%%%%%%%%%%%%%%%%%%%%%%%%%%%%%%%%%%%%%

%%%%%%%%%%%%%%%%%%%%%%%%%%%%%%%%%%%%%%%%%%%%%%%%%%%%%%%%%%%%%%%%%%%%%%%%%%%%%%%
\begin{table}[h!]
  \centering
  \begin{tabular}{|c||c|c|c|c|c|c|c|}
    \hline
    $D$ & NPT & NPT$3$ & ONPT$3$ & ONPT$4$ & NPT$5$\\
    \hline \hline
    10 &100\% & 19.54\% & 29.74\% & 100\% & 100\%\\
    \hline
    20 &100\% & 18.81\% & 29.10\% & 100\% & 100\%\\
    \hline
    30 &100\% & 18.51\% & 28.83\% & 100\% & 100\%\\
    \hline
    40 &100\% & 18.50\% & 28.92\% & 100\% & 100\%\\
    \hline \hline
  \end{tabular}
  \caption{%
    Fraction of (large) $D\times D$ states in the Hilbert-Schmidt distribution 
    (100,000 samples) that can be detected with the various criteria. Here, NPT 
    denotes the states violating the PPT criterion, NPT$n$ (NPT$3$, NPT$5$) 
    denote the states violating the $p_n$-PPT criterion in 
    Eq.~\eqref{eq:pnPPT}, and ONPT$n$ (ONPT$3$, ONPT$4$) denote the states 
    violating the $p_n$-OPPT criterion.
    Note that in Tables~\ref{tab:HS} and \ref{tab:HSx10}, the difference 
    between $100\%$ and $100.00\%$ is that the former is an exact value, while 
    the latter is an approximate value.
    In addition, see Ref.~\cite{Szymanski.etal2017} for an asymptotic behavior 
    ($D\to+\infty$) of the moments for all PPT states.
  }
  \label{tab:HSx10}
\end{table}
%%%%%%%%%%%%%%%%%%%%%%%%%%%%%%%%%%%%%%%%%%%%%%%%%%%%%%%%%%%%%%%%%%%%%%%%%%%%%%%

%%%%%%%%%%%%%%%%%%%%%%%%%%%%%%%%%%%%%%%%%%%%%%%%%%%%%%%%%%%%%%%%%%%%%%%%%%%%%
\section{Statistical analysis}\label{app:statistics}
%%%%%%%%%%%%%%%%%%%%%%%%%%%%%%%%%%%%%%%%%%%%%%%%%%%%%%%%%%%%%%%%%%%%%%%%%%%%%

In Refs.~\cite{Elben.etal2020,Huang.etal2020}, it was already shown that the 
PT-moments $p_n=\Tr[(\rho_{AB}^{T_A})^n]$ can be efficiently measured with the 
classical shadows and $U$-statistics. The statistical error is also explicitly 
analyzed for $n\le 3$. The extension the higher-order case is straightforward 
but tedious. To obtain a general result for arbitrary $n$, we focus on the high 
accuracy limit (i.e., the error $\varepsilon\ll 1$), and show that for any 
$N$-qubit state and an arbitrary bipartition $(A|B)$, a total of
\begin{equation}
  % M=\cO\left(\frac{n^22^Np_2^{n-1}}{\varepsilon^2\delta}\right)
  M\sim\frac{n^22^Np_2^{n-1}}{\varepsilon^2\delta}
  \label{eq:sample}
\end{equation}
copies of the state suffice to ensure that the estimated $\hat{p}_n$ obeys 
$\abs{\hat{p}_n-p_n}\le\varepsilon$ with probability at least $1-\delta$.  The 
$n^2p_2^{n-1}$ scaling in Eq.~\eqref{eq:sample} implies that measuring the 
higher-order PT-moments does not take many more measurements compared to the 
lower-order ones; moreover, it also shows that for fixed $N,n,\delta$ the error 
$\varepsilon$ is proportional to $p_2^{(n-1)/2}$.  This property is very 
important for tackling the problem that higher-order PT-moments decrease 
exponentially, i.e., $p_n\le p_2^{n/2}$ discussed in the main text.

Equation~\eqref{eq:sample} follows directly from the statistical analysis in 
Appendix~D of Ref.~\cite{Elben.etal2020}. For simplicity, we do not restate the 
data acquisition and processing protocol, but refer the readers to 
Ref.~\cite{Elben.etal2020,Huang.etal2020} for the details of the classical 
shadow formalism. Also, we employ similar notations to make the proof easier to 
follow. Similarly to Eqs.~(D16,\,D22) in Ref.~\cite{Elben.etal2020}, one can 
easily show that the variance of the estimator $\hat{p}_n$ reads
\begin{align}
  \label{eq:Varpn1}
  \Var(\hat{p}_n)=&{\binom{M}{n}}^{-1}\binom{n}{1}\binom{M-n}{n-1}
  \Var\left[\Tr\left(\left(\rho_{AB}^{T_A}\right)^{n-1}
  \hat{\rho}_{AB}^{T_A}\right)\right] + \cO\left(\frac{1}{M^2}\right)\\
  \label{eq:Varpn2}
  \le&\frac{n^2}{M}L + \cO\left(\frac{1}{M^2}\right),
\end{align}
where $L$ denotes the linear contribution
\begin{equation}
  L=\Var\left[\Tr\left(\left(\rho_{AB}^{T_A}\right)^{n-1} 
  \hat{\rho}_{AB}^{T_A}\right)\right],
\end{equation}
and its coefficient results from
\begin{equation}
  \frac{\binom{n}{1}\binom{M-n}{n-1}}{\binom{M}{n}}
  =\frac{\binom{n}{1}\binom{M-n}{n-1}}{\frac{M}{n}\binom{M-1}{n-1}}
  \le\frac{\binom{n}{1}}{\frac{M}{n}}=\frac{n^2}{M}.
  \label{eq:coeff}
\end{equation}
The high accuracy limit $\varepsilon\ll 1$ also implies that $M\gg 1$, thus we 
can ignore the higher-order terms in Eq.~\eqref{eq:Varpn2}. To estimate the 
linear contribution $L$, we set
\begin{equation}
  O=\left(\rho_{AB}^{T_A}\right)^{n-1},
\end{equation}
then
\begin{equation}
  L=\Var\left[\Tr\left(O\hat{\rho}_{AB}^{T_A}\right)\right]
  =\Var\left[\Tr\left(O^{T_A}\hat{\rho}_{AB}\right)\right]
  \le 2^N\Tr\big[(O^{T_A})^2\big],
\end{equation}
where the inequality follows from Eq.~(D7) in Ref.~\cite{Elben.etal2020}.  
Further,
\begin{equation}
  \Tr\big[(O^{T_A})^2\big]=\Tr\big[O^2\big]
  =\Tr\left[(\rho_{AB}^{T_A})^{2(n-1)}\right]
  \le\left[\Tr\big[(\rho_{AB}^{T_A})^{2}\big]\right]^{n-1}
  =p_2^{n-1},
\end{equation}
where we have used the relation that $\Tr(X^{n-1})\le[\Tr(X)]^{n-1}$ for any 
$X\ge 0$. Thus, we get that
\begin{equation}
  \Var(\hat{p}_n)\le\frac{n^22^Np_2^{n-1}}{M}+\cO\left(\frac{1}{M^2}\right).
\end{equation}
Then, Eq.~\eqref{eq:sample} follows directly from the Chebyshev inequality
\begin{equation}
  P\big(\abs{\hat{p}_n-p_n}\ge\varepsilon\big)
  \le\frac{\Var(\hat{p}_n)}{\varepsilon^2}.
\end{equation}

\twocolumngrid
\bibliography{QuantumInf}

\begin{thebibliography}{50}%
\makeatletter
\providecommand \@ifxundefined [1]{%
 \@ifx{#1\undefined}
}%
\providecommand \@ifnum [1]{%
 \ifnum #1\expandafter \@firstoftwo
 \else \expandafter \@secondoftwo
 \fi
}%
\providecommand \@ifx [1]{%
 \ifx #1\expandafter \@firstoftwo
 \else \expandafter \@secondoftwo
 \fi
}%
\providecommand \natexlab [1]{#1}%
\providecommand \enquote  [1]{``#1''}%
\providecommand \bibnamefont  [1]{#1}%
\providecommand \bibfnamefont [1]{#1}%
\providecommand \citenamefont [1]{#1}%
\providecommand \href@noop [0]{\@secondoftwo}%
\providecommand \href [0]{\begingroup \@sanitize@url \@href}%
\providecommand \@href[1]{\@@startlink{#1}\@@href}%
\providecommand \@@href[1]{\endgroup#1\@@endlink}%
\providecommand \@sanitize@url [0]{\catcode `\\12\catcode `\$12\catcode
  `\&12\catcode `\#12\catcode `\^12\catcode `\_12\catcode `\%12\relax}%
\providecommand \@@startlink[1]{}%
\providecommand \@@endlink[0]{}%
\providecommand \url  [0]{\begingroup\@sanitize@url \@url }%
\providecommand \@url [1]{\endgroup\@href {#1}{\urlprefix }}%
\providecommand \urlprefix  [0]{URL }%
\providecommand \Eprint [0]{\href }%
\providecommand \doibase [0]{https://doi.org/}%
\providecommand \selectlanguage [0]{\@gobble}%
\providecommand \bibinfo  [0]{\@secondoftwo}%
\providecommand \bibfield  [0]{\@secondoftwo}%
\providecommand \translation [1]{[#1]}%
\providecommand \BibitemOpen [0]{}%
\providecommand \bibitemStop [0]{}%
\providecommand \bibitemNoStop [0]{.\EOS\space}%
\providecommand \EOS [0]{\spacefactor3000\relax}%
\providecommand \BibitemShut  [1]{\csname bibitem#1\endcsname}%
\let\auto@bib@innerbib\@empty
%</preamble>
\bibitem [{\citenamefont {Preskill}(2018)}]{Preskill2018}%
  \BibitemOpen
  \bibfield  {author} {\bibinfo {author} {\bibfnamefont {J.}~\bibnamefont
  {Preskill}},\ }\bibfield  {title} {\bibinfo {title} {Quantum computing in the
  {NISQ} era and beyond},\ }\href {https://doi.org/10.22331/q-2018-08-06-79}
  {\bibfield  {journal} {\bibinfo  {journal} {{Quantum}}\ }\textbf {\bibinfo
  {volume} {2}},\ \bibinfo {pages} {79} (\bibinfo {year} {2018})}\BibitemShut
  {NoStop}%
\bibitem [{\citenamefont {Paris}\ and\ \citenamefont
  {\v{R}eh\'a\v{c}ek}(2004)}]{Paris.Rehacek2004}%
  \BibitemOpen
  \bibinfo {editor} {\bibfnamefont {M.}~\bibnamefont {Paris}}\ and\ \bibinfo
  {editor} {\bibfnamefont {J.}~\bibnamefont {\v{R}eh\'a\v{c}ek}},\ eds.,\
  \href@noop {} {\emph {\bibinfo {title} {Quantum State Estimation}}},\
  \bibinfo {series} {Lecture Notes in Physics}, Vol.\ \bibinfo {volume} {649}\
  (\bibinfo  {publisher} {Springer},\ \bibinfo {address} {Heidelberg},\
  \bibinfo {year} {2004})\BibitemShut {NoStop}%
\bibitem [{\citenamefont {Eisert}\ \emph {et~al.}(2020)\citenamefont {Eisert},
  \citenamefont {Hangleiter}, \citenamefont {Walk}, \citenamefont {Roth},
  \citenamefont {Markham}, \citenamefont {Parekh}, \citenamefont {Chabaud},\
  and\ \citenamefont {Kashefi}}]{Eisert.etal2020}%
  \BibitemOpen
  \bibfield  {author} {\bibinfo {author} {\bibfnamefont {J.}~\bibnamefont
  {Eisert}}, \bibinfo {author} {\bibfnamefont {D.}~\bibnamefont {Hangleiter}},
  \bibinfo {author} {\bibfnamefont {N.}~\bibnamefont {Walk}}, \bibinfo {author}
  {\bibfnamefont {I.}~\bibnamefont {Roth}}, \bibinfo {author} {\bibfnamefont
  {D.}~\bibnamefont {Markham}}, \bibinfo {author} {\bibfnamefont
  {R.}~\bibnamefont {Parekh}}, \bibinfo {author} {\bibfnamefont
  {U.}~\bibnamefont {Chabaud}},\ and\ \bibinfo {author} {\bibfnamefont
  {E.}~\bibnamefont {Kashefi}},\ }\bibfield  {title} {\bibinfo {title} {Quantum
  certification and benchmarking},\ }\href
  {https://doi.org/10.1038/s42254-020-0186-4} {\bibfield  {journal} {\bibinfo
  {journal} {Nat. Rev. Phys.}\ }\textbf {\bibinfo {volume} {2}},\ \bibinfo
  {pages} {382} (\bibinfo {year} {2020})}\BibitemShut {NoStop}%
\bibitem [{\citenamefont {Kliesch}\ and\ \citenamefont
  {Roth}(2021)}]{Kliesch.Roth2021}%
  \BibitemOpen
  \bibfield  {author} {\bibinfo {author} {\bibfnamefont {M.}~\bibnamefont
  {Kliesch}}\ and\ \bibinfo {author} {\bibfnamefont {I.}~\bibnamefont {Roth}},\
  }\bibfield  {title} {\bibinfo {title} {Theory of quantum system
  certification},\ }\href {https://doi.org/10.1103/PRXQuantum.2.010201}
  {\bibfield  {journal} {\bibinfo  {journal} {PRX Quantum}\ }\textbf {\bibinfo
  {volume} {2}},\ \bibinfo {pages} {010201} (\bibinfo {year}
  {2021})}\BibitemShut {NoStop}%
\bibitem [{\citenamefont {Horodecki}\ \emph {et~al.}(2009)\citenamefont
  {Horodecki}, \citenamefont {Horodecki}, \citenamefont {Horodecki},\ and\
  \citenamefont {Horodecki}}]{Horodecki.etal2009}%
  \BibitemOpen
  \bibfield  {author} {\bibinfo {author} {\bibfnamefont {R.}~\bibnamefont
  {Horodecki}}, \bibinfo {author} {\bibfnamefont {P.}~\bibnamefont
  {Horodecki}}, \bibinfo {author} {\bibfnamefont {M.}~\bibnamefont
  {Horodecki}},\ and\ \bibinfo {author} {\bibfnamefont {K.}~\bibnamefont
  {Horodecki}},\ }\bibfield  {title} {\bibinfo {title} {Quantum entanglement},\
  }\href {https://doi.org/10.1103/RevModPhys.81.865} {\bibfield  {journal}
  {\bibinfo  {journal} {Rev. Mod. Phys.}\ }\textbf {\bibinfo {volume} {81}},\
  \bibinfo {pages} {865} (\bibinfo {year} {2009})}\BibitemShut {NoStop}%
\bibitem [{\citenamefont {G{\"u}hne}\ and\ \citenamefont
  {T{\'o}th}(2009)}]{Guehne.Toth2009}%
  \BibitemOpen
  \bibfield  {author} {\bibinfo {author} {\bibfnamefont {O.}~\bibnamefont
  {G{\"u}hne}}\ and\ \bibinfo {author} {\bibfnamefont {G.}~\bibnamefont
  {T{\'o}th}},\ }\bibfield  {title} {\bibinfo {title} {Entanglement
  detection},\ }\href {https://doi.org/10.1016/j.physrep.2009.02.004}
  {\bibfield  {journal} {\bibinfo  {journal} {Phys. Rep.}\ }\textbf {\bibinfo
  {volume} {474}},\ \bibinfo {pages} {1} (\bibinfo {year} {2009})}\BibitemShut
  {NoStop}%
\bibitem [{\citenamefont {{Friis}}\ \emph {et~al.}(2019)\citenamefont
  {{Friis}}, \citenamefont {{Vitagliano}}, \citenamefont {{Malik}},\ and\
  \citenamefont {{Huber}}}]{Friis.etal2019}%
  \BibitemOpen
  \bibfield  {author} {\bibinfo {author} {\bibfnamefont {N.}~\bibnamefont
  {{Friis}}}, \bibinfo {author} {\bibfnamefont {G.}~\bibnamefont
  {{Vitagliano}}}, \bibinfo {author} {\bibfnamefont {M.}~\bibnamefont
  {{Malik}}},\ and\ \bibinfo {author} {\bibfnamefont {M.}~\bibnamefont
  {{Huber}}},\ }\bibfield  {title} {\bibinfo {title} {Entanglement
  certification from theory to experiment},\ }\href
  {https://doi.org/10.1038/s42254-018-0003-5} {\bibfield  {journal} {\bibinfo
  {journal} {Nat. Rev. Phys.}\ }\textbf {\bibinfo {volume} {1}},\ \bibinfo
  {pages} {72} (\bibinfo {year} {2019})}\BibitemShut {NoStop}%
\bibitem [{\citenamefont {Tran}\ \emph {et~al.}(2015)\citenamefont {Tran},
  \citenamefont {Daki\'{c}}, \citenamefont {Arnault}, \citenamefont
  {Laskowski},\ and\ \citenamefont {Paterek}}]{Tran.etal2015}%
  \BibitemOpen
  \bibfield  {author} {\bibinfo {author} {\bibfnamefont {M.~C.}\ \bibnamefont
  {Tran}}, \bibinfo {author} {\bibfnamefont {B.}~\bibnamefont {Daki\'{c}}},
  \bibinfo {author} {\bibfnamefont {F.}~\bibnamefont {Arnault}}, \bibinfo
  {author} {\bibfnamefont {W.}~\bibnamefont {Laskowski}},\ and\ \bibinfo
  {author} {\bibfnamefont {T.}~\bibnamefont {Paterek}},\ }\bibfield  {title}
  {\bibinfo {title} {Quantum entanglement from random measurements},\ }\href
  {https://doi.org/10.1103/PhysRevA.92.050301} {\bibfield  {journal} {\bibinfo
  {journal} {Phys. Rev. A}\ }\textbf {\bibinfo {volume} {92}},\ \bibinfo
  {pages} {050301(R)} (\bibinfo {year} {2015})}\BibitemShut {NoStop}%
\bibitem [{\citenamefont {Tran}\ \emph {et~al.}(2016)\citenamefont {Tran},
  \citenamefont {Daki\ifmmode~\acute{c}\else \'{c}\fi{}}, \citenamefont
  {Laskowski},\ and\ \citenamefont {Paterek}}]{Tran.etal2016}%
  \BibitemOpen
  \bibfield  {author} {\bibinfo {author} {\bibfnamefont {M.~C.}\ \bibnamefont
  {Tran}}, \bibinfo {author} {\bibfnamefont {B.}~\bibnamefont
  {Daki\ifmmode~\acute{c}\else \'{c}\fi{}}}, \bibinfo {author} {\bibfnamefont
  {W.}~\bibnamefont {Laskowski}},\ and\ \bibinfo {author} {\bibfnamefont
  {T.}~\bibnamefont {Paterek}},\ }\bibfield  {title} {\bibinfo {title}
  {Correlations between outcomes of random measurements},\ }\href
  {https://doi.org/10.1103/PhysRevA.94.042302} {\bibfield  {journal} {\bibinfo
  {journal} {Phys. Rev. A}\ }\textbf {\bibinfo {volume} {94}},\ \bibinfo
  {pages} {042302} (\bibinfo {year} {2016})}\BibitemShut {NoStop}%
\bibitem [{\citenamefont {{van Enk}}\ and\ \citenamefont
  {Beenakker}(2012)}]{vanEnk.Beenakker2012}%
  \BibitemOpen
  \bibfield  {author} {\bibinfo {author} {\bibfnamefont {S.~J.}\ \bibnamefont
  {{van Enk}}}\ and\ \bibinfo {author} {\bibfnamefont {C.~W.~J.}\ \bibnamefont
  {Beenakker}},\ }\bibfield  {title} {\bibinfo {title} {Measuring
  $\mathrm{Tr}{\ensuremath{\rho}}^{n}$ on single copies of $\ensuremath{\rho}$
  using random measurements},\ }\href
  {https://doi.org/10.1103/PhysRevLett.108.110503} {\bibfield  {journal}
  {\bibinfo  {journal} {Phys. Rev. Lett.}\ }\textbf {\bibinfo {volume} {108}},\
  \bibinfo {pages} {110503} (\bibinfo {year} {2012})}\BibitemShut {NoStop}%
\bibitem [{\citenamefont {Elben}\ \emph {et~al.}(2019)\citenamefont {Elben},
  \citenamefont {Vermersch}, \citenamefont {Roos},\ and\ \citenamefont
  {Zoller}}]{Elben.etal2019}%
  \BibitemOpen
  \bibfield  {author} {\bibinfo {author} {\bibfnamefont {A.}~\bibnamefont
  {Elben}}, \bibinfo {author} {\bibfnamefont {B.}~\bibnamefont {Vermersch}},
  \bibinfo {author} {\bibfnamefont {C.~F.}\ \bibnamefont {Roos}},\ and\
  \bibinfo {author} {\bibfnamefont {P.}~\bibnamefont {Zoller}},\ }\bibfield
  {title} {\bibinfo {title} {Statistical correlations between locally
  randomized measurements: A toolbox for probing entanglement in many-body
  quantum states},\ }\href {https://doi.org/10.1103/PhysRevA.99.052323}
  {\bibfield  {journal} {\bibinfo  {journal} {Phys. Rev. A}\ }\textbf {\bibinfo
  {volume} {99}},\ \bibinfo {pages} {052323} (\bibinfo {year}
  {2019})}\BibitemShut {NoStop}%
\bibitem [{\citenamefont {Brydges}\ \emph {et~al.}(2019)\citenamefont
  {Brydges}, \citenamefont {Elben}, \citenamefont {Jurcevic}, \citenamefont
  {Vermersch}, \citenamefont {Maier}, \citenamefont {Lanyon}, \citenamefont
  {Zoller}, \citenamefont {Blatt},\ and\ \citenamefont
  {Roos}}]{Brydges.etal2019}%
  \BibitemOpen
  \bibfield  {author} {\bibinfo {author} {\bibfnamefont {T.}~\bibnamefont
  {Brydges}}, \bibinfo {author} {\bibfnamefont {A.}~\bibnamefont {Elben}},
  \bibinfo {author} {\bibfnamefont {P.}~\bibnamefont {Jurcevic}}, \bibinfo
  {author} {\bibfnamefont {B.}~\bibnamefont {Vermersch}}, \bibinfo {author}
  {\bibfnamefont {C.}~\bibnamefont {Maier}}, \bibinfo {author} {\bibfnamefont
  {B.~P.}\ \bibnamefont {Lanyon}}, \bibinfo {author} {\bibfnamefont
  {P.}~\bibnamefont {Zoller}}, \bibinfo {author} {\bibfnamefont
  {R.}~\bibnamefont {Blatt}},\ and\ \bibinfo {author} {\bibfnamefont {C.~F.}\
  \bibnamefont {Roos}},\ }\bibfield  {title} {\bibinfo {title} {Probing
  r{\'e}nyi entanglement entropy via randomized measurements},\ }\href
  {https://doi.org/10.1126/science.aau4963} {\bibfield  {journal} {\bibinfo
  {journal} {Science}\ }\textbf {\bibinfo {volume} {364}},\ \bibinfo {pages}
  {260} (\bibinfo {year} {2019})}\BibitemShut {NoStop}%
\bibitem [{\citenamefont {Ketterer}\ \emph {et~al.}(2019)\citenamefont
  {Ketterer}, \citenamefont {Wyderka},\ and\ \citenamefont
  {G\"uhne}}]{Ketterer.etal2019}%
  \BibitemOpen
  \bibfield  {author} {\bibinfo {author} {\bibfnamefont {A.}~\bibnamefont
  {Ketterer}}, \bibinfo {author} {\bibfnamefont {N.}~\bibnamefont {Wyderka}},\
  and\ \bibinfo {author} {\bibfnamefont {O.}~\bibnamefont {G\"uhne}},\
  }\bibfield  {title} {\bibinfo {title} {Characterizing multipartite
  entanglement with moments of random correlations},\ }\href
  {https://doi.org/10.1103/PhysRevLett.122.120505} {\bibfield  {journal}
  {\bibinfo  {journal} {Phys. Rev. Lett.}\ }\textbf {\bibinfo {volume} {122}},\
  \bibinfo {pages} {120505} (\bibinfo {year} {2019})}\BibitemShut {NoStop}%
\bibitem [{\citenamefont {Ketterer}\ \emph {et~al.}(2020)\citenamefont
  {Ketterer}, \citenamefont {Wyderka},\ and\ \citenamefont
  {G{\"{u}}hne}}]{Ketterer.etal2020}%
  \BibitemOpen
  \bibfield  {author} {\bibinfo {author} {\bibfnamefont {A.}~\bibnamefont
  {Ketterer}}, \bibinfo {author} {\bibfnamefont {N.}~\bibnamefont {Wyderka}},\
  and\ \bibinfo {author} {\bibfnamefont {O.}~\bibnamefont {G{\"{u}}hne}},\
  }\bibfield  {title} {\bibinfo {title} {Entanglement characterization using
  quantum designs},\ }\href {https://doi.org/10.22331/q-2020-09-16-325}
  {\bibfield  {journal} {\bibinfo  {journal} {{Quantum}}\ }\textbf {\bibinfo
  {volume} {4}},\ \bibinfo {pages} {325} (\bibinfo {year} {2020})}\BibitemShut
  {NoStop}%
\bibitem [{\citenamefont {{Ketterer}}\ \emph {et~al.}()\citenamefont
  {{Ketterer}}, \citenamefont {{Imai}}, \citenamefont {{Wyderka}},\ and\
  \citenamefont {{G{\"u}hne}}}]{Ketterer.etal2020b}%
  \BibitemOpen
  \bibfield  {author} {\bibinfo {author} {\bibfnamefont {A.}~\bibnamefont
  {{Ketterer}}}, \bibinfo {author} {\bibfnamefont {S.}~\bibnamefont {{Imai}}},
  \bibinfo {author} {\bibfnamefont {N.}~\bibnamefont {{Wyderka}}},\ and\
  \bibinfo {author} {\bibfnamefont {O.}~\bibnamefont {{G{\"u}hne}}},\
  }\href@noop {} {\bibinfo {title} {Certifying multiparticle entanglement with
  randomized measurements}},\ \Eprint {https://arxiv.org/abs/2012.12176}
  {arXiv:2012.12176} \BibitemShut {NoStop}%
\bibitem [{\citenamefont {{Knips}}\ \emph {et~al.}(2020)\citenamefont
  {{Knips}}, \citenamefont {{Dziewior}}, \citenamefont {{K{\l}obus}},
  \citenamefont {{Laskowski}}, \citenamefont {{Paterek}}, \citenamefont
  {{Shadbolt}}, \citenamefont {{Weinfurter}},\ and\ \citenamefont
  {{Meinecke}}}]{Knips.etal2020}%
  \BibitemOpen
  \bibfield  {author} {\bibinfo {author} {\bibfnamefont {L.}~\bibnamefont
  {{Knips}}}, \bibinfo {author} {\bibfnamefont {J.}~\bibnamefont {{Dziewior}}},
  \bibinfo {author} {\bibfnamefont {W.}~\bibnamefont {{K{\l}obus}}}, \bibinfo
  {author} {\bibfnamefont {W.}~\bibnamefont {{Laskowski}}}, \bibinfo {author}
  {\bibfnamefont {T.}~\bibnamefont {{Paterek}}}, \bibinfo {author}
  {\bibfnamefont {P.~J.}\ \bibnamefont {{Shadbolt}}}, \bibinfo {author}
  {\bibfnamefont {H.}~\bibnamefont {{Weinfurter}}},\ and\ \bibinfo {author}
  {\bibfnamefont {J.~D.~A.}\ \bibnamefont {{Meinecke}}},\ }\bibfield  {title}
  {\bibinfo {title} {Multipartite entanglement analysis from random
  correlations},\ }\href {https://doi.org/10.1038/s41534-020-0281-5} {\bibfield
   {journal} {\bibinfo  {journal} {npj Quantum Inf.}\ }\textbf {\bibinfo
  {volume} {6}},\ \bibinfo {pages} {51} (\bibinfo {year} {2020})}\BibitemShut
  {NoStop}%
\bibitem [{\citenamefont {Imai}\ \emph {et~al.}(2021)\citenamefont {Imai},
  \citenamefont {Wyderka}, \citenamefont {Ketterer},\ and\ \citenamefont
  {G\"uhne}}]{Imai.etal2021}%
  \BibitemOpen
  \bibfield  {author} {\bibinfo {author} {\bibfnamefont {S.}~\bibnamefont
  {Imai}}, \bibinfo {author} {\bibfnamefont {N.}~\bibnamefont {Wyderka}},
  \bibinfo {author} {\bibfnamefont {A.}~\bibnamefont {Ketterer}},\ and\
  \bibinfo {author} {\bibfnamefont {O.}~\bibnamefont {G\"uhne}},\ }\bibfield
  {title} {\bibinfo {title} {Bound entanglement from randomized measurements},\
  }\href {https://doi.org/10.1103/PhysRevLett.126.150501} {\bibfield  {journal}
  {\bibinfo  {journal} {Phys. Rev. Lett.}\ }\textbf {\bibinfo {volume} {126}},\
  \bibinfo {pages} {150501} (\bibinfo {year} {2021})}\BibitemShut {NoStop}%
\bibitem [{\citenamefont {Peres}(1996)}]{Peres1996}%
  \BibitemOpen
  \bibfield  {author} {\bibinfo {author} {\bibfnamefont {A.}~\bibnamefont
  {Peres}},\ }\bibfield  {title} {\bibinfo {title} {Separability criterion for
  density matrices},\ }\href {https://doi.org/10.1103/PhysRevLett.77.1413}
  {\bibfield  {journal} {\bibinfo  {journal} {Phys. Rev. Lett.}\ }\textbf
  {\bibinfo {volume} {77}},\ \bibinfo {pages} {1413} (\bibinfo {year}
  {1996})}\BibitemShut {NoStop}%
\bibitem [{\citenamefont {Gray}\ \emph {et~al.}(2018)\citenamefont {Gray},
  \citenamefont {Banchi}, \citenamefont {Bayat},\ and\ \citenamefont
  {Bose}}]{Gray.etal2018}%
  \BibitemOpen
  \bibfield  {author} {\bibinfo {author} {\bibfnamefont {J.}~\bibnamefont
  {Gray}}, \bibinfo {author} {\bibfnamefont {L.}~\bibnamefont {Banchi}},
  \bibinfo {author} {\bibfnamefont {A.}~\bibnamefont {Bayat}},\ and\ \bibinfo
  {author} {\bibfnamefont {S.}~\bibnamefont {Bose}},\ }\bibfield  {title}
  {\bibinfo {title} {Machine-learning-assisted many-body entanglement
  measurement},\ }\href {https://doi.org/10.1103/PhysRevLett.121.150503}
  {\bibfield  {journal} {\bibinfo  {journal} {Phys. Rev. Lett.}\ }\textbf
  {\bibinfo {volume} {121}},\ \bibinfo {pages} {150503} (\bibinfo {year}
  {2018})}\BibitemShut {NoStop}%
\bibitem [{\citenamefont {Elben}\ \emph {et~al.}(2020)\citenamefont {Elben},
  \citenamefont {Kueng}, \citenamefont {Huang}, \citenamefont {van Bijnen},
  \citenamefont {Kokail}, \citenamefont {Dalmonte}, \citenamefont {Calabrese},
  \citenamefont {Kraus}, \citenamefont {Preskill}, \citenamefont {Zoller},\
  and\ \citenamefont {Vermersch}}]{Elben.etal2020}%
  \BibitemOpen
  \bibfield  {author} {\bibinfo {author} {\bibfnamefont {A.}~\bibnamefont
  {Elben}}, \bibinfo {author} {\bibfnamefont {R.}~\bibnamefont {Kueng}},
  \bibinfo {author} {\bibfnamefont {H.-Y.~R.}\ \bibnamefont {Huang}}, \bibinfo
  {author} {\bibfnamefont {R.}~\bibnamefont {van Bijnen}}, \bibinfo {author}
  {\bibfnamefont {C.}~\bibnamefont {Kokail}}, \bibinfo {author} {\bibfnamefont
  {M.}~\bibnamefont {Dalmonte}}, \bibinfo {author} {\bibfnamefont
  {P.}~\bibnamefont {Calabrese}}, \bibinfo {author} {\bibfnamefont
  {B.}~\bibnamefont {Kraus}}, \bibinfo {author} {\bibfnamefont
  {J.}~\bibnamefont {Preskill}}, \bibinfo {author} {\bibfnamefont
  {P.}~\bibnamefont {Zoller}},\ and\ \bibinfo {author} {\bibfnamefont
  {B.}~\bibnamefont {Vermersch}},\ }\bibfield  {title} {\bibinfo {title}
  {Mixed-state entanglement from local randomized measurements},\ }\href
  {https://doi.org/10.1103/PhysRevLett.125.200501} {\bibfield  {journal}
  {\bibinfo  {journal} {Phys. Rev. Lett.}\ }\textbf {\bibinfo {volume} {125}},\
  \bibinfo {pages} {200501} (\bibinfo {year} {2020})}\BibitemShut {NoStop}%
\bibitem [{\citenamefont {Zhou}\ \emph {et~al.}(2020)\citenamefont {Zhou},
  \citenamefont {Zeng},\ and\ \citenamefont {Liu}}]{Zhou.etal2020}%
  \BibitemOpen
  \bibfield  {author} {\bibinfo {author} {\bibfnamefont {Y.}~\bibnamefont
  {Zhou}}, \bibinfo {author} {\bibfnamefont {P.}~\bibnamefont {Zeng}},\ and\
  \bibinfo {author} {\bibfnamefont {Z.}~\bibnamefont {Liu}},\ }\bibfield
  {title} {\bibinfo {title} {Single-copies estimation of entanglement
  negativity},\ }\href {https://doi.org/10.1103/PhysRevLett.125.200502}
  {\bibfield  {journal} {\bibinfo  {journal} {Phys. Rev. Lett.}\ }\textbf
  {\bibinfo {volume} {125}},\ \bibinfo {pages} {200502} (\bibinfo {year}
  {2020})}\BibitemShut {NoStop}%
\bibitem [{\citenamefont {{Huang}}\ \emph {et~al.}(2020)\citenamefont
  {{Huang}}, \citenamefont {{Kueng}},\ and\ \citenamefont
  {{Preskill}}}]{Huang.etal2020}%
  \BibitemOpen
  \bibfield  {author} {\bibinfo {author} {\bibfnamefont {H.-Y.}\ \bibnamefont
  {{Huang}}}, \bibinfo {author} {\bibfnamefont {R.}~\bibnamefont {{Kueng}}},\
  and\ \bibinfo {author} {\bibfnamefont {J.}~\bibnamefont {{Preskill}}},\
  }\bibfield  {title} {\bibinfo {title} {Predicting many properties of a
  quantum system from very few measurements},\ }\href
  {https://doi.org/10.1038/s41567-020-0932-7} {\bibfield  {journal} {\bibinfo
  {journal} {Nat. Phys.}\ }\textbf {\bibinfo {volume} {16}},\ \bibinfo {pages}
  {1050} (\bibinfo {year} {2020})}\BibitemShut {NoStop}%
\bibitem [{\citenamefont {Ekert}\ \emph {et~al.}(2002)\citenamefont {Ekert},
  \citenamefont {Alves}, \citenamefont {Oi}, \citenamefont {Horodecki},
  \citenamefont {Horodecki},\ and\ \citenamefont {Kwek}}]{Ekert.etal2002}%
  \BibitemOpen
  \bibfield  {author} {\bibinfo {author} {\bibfnamefont {A.~K.}\ \bibnamefont
  {Ekert}}, \bibinfo {author} {\bibfnamefont {C.~M.}\ \bibnamefont {Alves}},
  \bibinfo {author} {\bibfnamefont {D.~K.~L.}\ \bibnamefont {Oi}}, \bibinfo
  {author} {\bibfnamefont {M.}~\bibnamefont {Horodecki}}, \bibinfo {author}
  {\bibfnamefont {P.}~\bibnamefont {Horodecki}},\ and\ \bibinfo {author}
  {\bibfnamefont {L.~C.}\ \bibnamefont {Kwek}},\ }\bibfield  {title} {\bibinfo
  {title} {Direct estimations of linear and nonlinear functionals of a quantum
  state},\ }\href {https://doi.org/10.1103/PhysRevLett.88.217901} {\bibfield
  {journal} {\bibinfo  {journal} {Phys. Rev. Lett.}\ }\textbf {\bibinfo
  {volume} {88}},\ \bibinfo {pages} {217901} (\bibinfo {year}
  {2002})}\BibitemShut {NoStop}%
\bibitem [{\citenamefont {Bohnet-Waldraff}\ \emph {et~al.}(2017)\citenamefont
  {Bohnet-Waldraff}, \citenamefont {Braun},\ and\ \citenamefont
  {Giraud}}]{BohnetWaldraff.etal2017}%
  \BibitemOpen
  \bibfield  {author} {\bibinfo {author} {\bibfnamefont {F.}~\bibnamefont
  {Bohnet-Waldraff}}, \bibinfo {author} {\bibfnamefont {D.}~\bibnamefont
  {Braun}},\ and\ \bibinfo {author} {\bibfnamefont {O.}~\bibnamefont
  {Giraud}},\ }\bibfield  {title} {\bibinfo {title} {Entanglement and the
  truncated moment problem},\ }\href
  {https://doi.org/10.1103/PhysRevA.96.032312} {\bibfield  {journal} {\bibinfo
  {journal} {Phys. Rev. A}\ }\textbf {\bibinfo {volume} {96}},\ \bibinfo
  {pages} {032312} (\bibinfo {year} {2017})}\BibitemShut {NoStop}%
\bibitem [{\citenamefont {Milazzo}\ \emph {et~al.}(2020)\citenamefont
  {Milazzo}, \citenamefont {Braun},\ and\ \citenamefont
  {Giraud}}]{Milazzo.etal2020}%
  \BibitemOpen
  \bibfield  {author} {\bibinfo {author} {\bibfnamefont {N.}~\bibnamefont
  {Milazzo}}, \bibinfo {author} {\bibfnamefont {D.}~\bibnamefont {Braun}},\
  and\ \bibinfo {author} {\bibfnamefont {O.}~\bibnamefont {Giraud}},\
  }\bibfield  {title} {\bibinfo {title} {Truncated moment sequences and a
  solution to the channel separability problem},\ }\href
  {https://doi.org/10.1103/PhysRevA.102.052406} {\bibfield  {journal} {\bibinfo
   {journal} {Phys. Rev. A}\ }\textbf {\bibinfo {volume} {102}},\ \bibinfo
  {pages} {052406} (\bibinfo {year} {2020})}\BibitemShut {NoStop}%
\bibitem [{\citenamefont {Shchukin}\ and\ \citenamefont
  {Vogel}(2005)}]{Shchukin.Vogel2005}%
  \BibitemOpen
  \bibfield  {author} {\bibinfo {author} {\bibfnamefont {E.}~\bibnamefont
  {Shchukin}}\ and\ \bibinfo {author} {\bibfnamefont {W.}~\bibnamefont
  {Vogel}},\ }\bibfield  {title} {\bibinfo {title} {Inseparability criteria for
  continuous bipartite quantum states},\ }\href
  {https://doi.org/10.1103/PhysRevLett.95.230502} {\bibfield  {journal}
  {\bibinfo  {journal} {Phys. Rev. Lett.}\ }\textbf {\bibinfo {volume} {95}},\
  \bibinfo {pages} {230502} (\bibinfo {year} {2005})}\BibitemShut {NoStop}%
\bibitem [{\citenamefont {Akhiezer}(1965)}]{Akhiezer1965}%
  \BibitemOpen
  \bibfield  {author} {\bibinfo {author} {\bibfnamefont {N.~I.}\ \bibnamefont
  {Akhiezer}},\ }\href@noop {} {\emph {\bibinfo {title} {The classical moment
  problem and some related questions in analysis}}}\ (\bibinfo  {publisher}
  {Oliver \& Boyd},\ \bibinfo {year} {1965})\BibitemShut {NoStop}%
\bibitem [{\citenamefont {Lasserre}(2010)}]{Lasserre2010}%
  \BibitemOpen
  \bibfield  {author} {\bibinfo {author} {\bibfnamefont {J.-B.}\ \bibnamefont
  {Lasserre}},\ }\href@noop {} {\emph {\bibinfo {title} {Moments, Positive
  Polynomials and Their Applications}}}\ (\bibinfo  {publisher} {World
  Scientific},\ \bibinfo {year} {2010})\BibitemShut {NoStop}%
\bibitem [{\citenamefont {Schm{\"u}dgen}(2017)}]{Schmuedgen2017}%
  \BibitemOpen
  \bibfield  {author} {\bibinfo {author} {\bibfnamefont {K.}~\bibnamefont
  {Schm{\"u}dgen}},\ }\href@noop {} {\emph {\bibinfo {title} {The Moment
  Problem}}}\ (\bibinfo  {publisher} {Springer, Cham},\ \bibinfo {year}
  {2017})\BibitemShut {NoStop}%
\bibitem [{\citenamefont {Lasserre}(2001)}]{Lasserre2001}%
  \BibitemOpen
  \bibfield  {author} {\bibinfo {author} {\bibfnamefont {J.~B.}\ \bibnamefont
  {Lasserre}},\ }\bibfield  {title} {\bibinfo {title} {Global optimization with
  polynomials and the problem of moments},\ }\href
  {https://doi.org/10.1137/S1052623400366802} {\bibfield  {journal} {\bibinfo
  {journal} {SIAM J. Optim.}\ }\textbf {\bibinfo {volume} {11}},\ \bibinfo
  {pages} {796} (\bibinfo {year} {2001})}\BibitemShut {NoStop}%
\bibitem [{\citenamefont {Parrilo}(2000)}]{Parrilo2000}%
  \BibitemOpen
  \bibfield  {author} {\bibinfo {author} {\bibfnamefont {P.~A.}\ \bibnamefont
  {Parrilo}},\ }\emph {\bibinfo {title} {Structured semidefinite programs and
  semialgebraic geometry methods in robustness and optimization}},\ \href
  {https://doi.org/10.7907/2K6Y-CH43} {Ph.D. thesis},\ \bibinfo  {school}
  {California Institute of Technology} (\bibinfo {year} {2000})\BibitemShut
  {NoStop}%
\bibitem [{\citenamefont {{De las Cuevas}}\ \emph {et~al.}(2020)\citenamefont
  {{De las Cuevas}}, \citenamefont {{Fritz}},\ and\ \citenamefont
  {{Netzer}}}]{DelasCuevas.etal2020}%
  \BibitemOpen
  \bibfield  {author} {\bibinfo {author} {\bibfnamefont {G.}~\bibnamefont {{De
  las Cuevas}}}, \bibinfo {author} {\bibfnamefont {T.}~\bibnamefont
  {{Fritz}}},\ and\ \bibinfo {author} {\bibfnamefont {T.}~\bibnamefont
  {{Netzer}}},\ }\bibfield  {title} {\bibinfo {title} {Optimal bounds on the
  positivity of a matrix from a few moments},\ }\href
  {https://doi.org/10.1007/s00220-020-03720-5} {\bibfield  {journal} {\bibinfo
  {journal} {Comm. Math. Phys.}\ }\textbf {\bibinfo {volume} {375}},\ \bibinfo
  {pages} {105} (\bibinfo {year} {2020})}\BibitemShut {NoStop}%
\bibitem [{\citenamefont {Abel}(1824)}]{Abel1824}%
  \BibitemOpen
  \bibfield  {author} {\bibinfo {author} {\bibfnamefont {N.~H.}\ \bibnamefont
  {Abel}},\ }\href@noop {} {\emph {\bibinfo {title} {M{\'e}moire sur les
  {\'e}quations alg{\'e}briques, o{\`u} on demontre l'impossibilit{\'e} de la
  r{\'e}solution de l'{\'e}quation g{\'e}n{\'e}rale du cinqui{\`e}me
  d{\'e}gr{\'e}}}}\ (\bibinfo  {publisher} {De l'imprimerie de Groendahl,
  Christiania},\ \bibinfo {year} {1824})\BibitemShut {NoStop}%
\bibitem [{Sup()}]{SuppArxiv}%
  \BibitemOpen
  \href@noop {} {}\bibinfo {note} {The computer code is included in the source
  files of this {arXiv} submission.}\BibitemShut {Stop}%
\bibitem [{\citenamefont {{\.Z}yczkowski}\ \emph {et~al.}(1998)\citenamefont
  {{\.Z}yczkowski}, \citenamefont {Horodecki}, \citenamefont {Sanpera},\ and\
  \citenamefont {Lewenstein}}]{Zyczkowski.etal1998}%
  \BibitemOpen
  \bibfield  {author} {\bibinfo {author} {\bibfnamefont {K.}~\bibnamefont
  {{\.Z}yczkowski}}, \bibinfo {author} {\bibfnamefont {P.}~\bibnamefont
  {Horodecki}}, \bibinfo {author} {\bibfnamefont {A.}~\bibnamefont {Sanpera}},\
  and\ \bibinfo {author} {\bibfnamefont {M.}~\bibnamefont {Lewenstein}},\
  }\bibfield  {title} {\bibinfo {title} {Volume of the set of separable
  states},\ }\href {https://doi.org/10.1103/PhysRevA.58.883} {\bibfield
  {journal} {\bibinfo  {journal} {Phys. Rev. A}\ }\textbf {\bibinfo {volume}
  {58}},\ \bibinfo {pages} {883} (\bibinfo {year} {1998})}\BibitemShut
  {NoStop}%
\bibitem [{\citenamefont {Vidal}\ and\ \citenamefont
  {Werner}(2002)}]{Vidal.Werner2002}%
  \BibitemOpen
  \bibfield  {author} {\bibinfo {author} {\bibfnamefont {G.}~\bibnamefont
  {Vidal}}\ and\ \bibinfo {author} {\bibfnamefont {R.~F.}\ \bibnamefont
  {Werner}},\ }\bibfield  {title} {\bibinfo {title} {Computable measure of
  entanglement},\ }\href {https://doi.org/10.1103/PhysRevA.65.032314}
  {\bibfield  {journal} {\bibinfo  {journal} {Phys. Rev. A}\ }\textbf {\bibinfo
  {volume} {65}},\ \bibinfo {pages} {032314} (\bibinfo {year}
  {2002})}\BibitemShut {NoStop}%
\bibitem [{\citenamefont {{{\.Z}yczkowski}}\ \emph {et~al.}(2011)\citenamefont
  {{{\.Z}yczkowski}}, \citenamefont {{Penson}}, \citenamefont {{Nechita}},\
  and\ \citenamefont {{Collins}}}]{Zyczkowski.etal2011}%
  \BibitemOpen
  \bibfield  {author} {\bibinfo {author} {\bibfnamefont {K.}~\bibnamefont
  {{{\.Z}yczkowski}}}, \bibinfo {author} {\bibfnamefont {K.~A.}\ \bibnamefont
  {{Penson}}}, \bibinfo {author} {\bibfnamefont {I.}~\bibnamefont
  {{Nechita}}},\ and\ \bibinfo {author} {\bibfnamefont {B.}~\bibnamefont
  {{Collins}}},\ }\bibfield  {title} {\bibinfo {title} {{Generating random
  density matrices}},\ }\href {https://doi.org/10.1063/1.3595693} {\bibfield
  {journal} {\bibinfo  {journal} {J. Math. Phys.}\ }\textbf {\bibinfo {volume}
  {52}},\ \bibinfo {pages} {062201} (\bibinfo {year} {2011})}\BibitemShut
  {NoStop}%
\bibitem [{\citenamefont {Weilenmann}\ \emph {et~al.}(2020)\citenamefont
  {Weilenmann}, \citenamefont {Dive}, \citenamefont {Trillo}, \citenamefont
  {Aguilar},\ and\ \citenamefont {Navascu\'es}}]{Weilenmann.etal2020}%
  \BibitemOpen
  \bibfield  {author} {\bibinfo {author} {\bibfnamefont {M.}~\bibnamefont
  {Weilenmann}}, \bibinfo {author} {\bibfnamefont {B.}~\bibnamefont {Dive}},
  \bibinfo {author} {\bibfnamefont {D.}~\bibnamefont {Trillo}}, \bibinfo
  {author} {\bibfnamefont {E.~A.}\ \bibnamefont {Aguilar}},\ and\ \bibinfo
  {author} {\bibfnamefont {M.}~\bibnamefont {Navascu\'es}},\ }\bibfield
  {title} {\bibinfo {title} {Entanglement detection beyond measuring
  fidelities},\ }\href {https://doi.org/10.1103/PhysRevLett.124.200502}
  {\bibfield  {journal} {\bibinfo  {journal} {Phys. Rev. Lett.}\ }\textbf
  {\bibinfo {volume} {124}},\ \bibinfo {pages} {200502} (\bibinfo {year}
  {2020})},\ \bibinfo {note} {{Erratum:
  \href{https://doi.org/10.1103/PhysRevLett.125.159903}{Phys. Rev. Lett.
  \textbf{125}, 159903(E) (2020)}}}\BibitemShut {NoStop}%
\bibitem [{\citenamefont {G\"uhne}\ \emph {et~al.}(2021)\citenamefont
  {G\"uhne}, \citenamefont {Mao},\ and\ \citenamefont {Yu}}]{Guehne.etal2021}%
  \BibitemOpen
  \bibfield  {author} {\bibinfo {author} {\bibfnamefont {O.}~\bibnamefont
  {G\"uhne}}, \bibinfo {author} {\bibfnamefont {Y.}~\bibnamefont {Mao}},\ and\
  \bibinfo {author} {\bibfnamefont {X.-D.}\ \bibnamefont {Yu}},\ }\bibfield
  {title} {\bibinfo {title} {Geometry of faithful entanglement},\ }\href
  {https://doi.org/10.1103/PhysRevLett.126.140503} {\bibfield  {journal}
  {\bibinfo  {journal} {Phys. Rev. Lett.}\ }\textbf {\bibinfo {volume} {126}},\
  \bibinfo {pages} {140503} (\bibinfo {year} {2021})}\BibitemShut {NoStop}%
\bibitem [{Note1()}]{Note1}%
  \BibitemOpen
  \bibinfo {note} {This follows directly from the Schur convexity \cite
  {Bhatia1997} of $f(\protect \bm {y})=\DOTSB \sum@ \slimits@ _iy_i^{n/2}$ for
  $y_i\in [0,+\infty )$.}\BibitemShut {Stop}%
\bibitem [{\citenamefont {{van Dam}}\ \emph {et~al.}(2005)\citenamefont {{van
  Dam}}, \citenamefont {{Gill}},\ and\ \citenamefont
  {{Grunwald}}}]{vanDam.etal2005}%
  \BibitemOpen
  \bibfield  {author} {\bibinfo {author} {\bibfnamefont {W.}~\bibnamefont {{van
  Dam}}}, \bibinfo {author} {\bibfnamefont {R.~D.}\ \bibnamefont {{Gill}}},\
  and\ \bibinfo {author} {\bibfnamefont {P.~D.}\ \bibnamefont {{Grunwald}}},\
  }\bibfield  {title} {\bibinfo {title} {The statistical strength of
  nonlocality proofs},\ }\href {https://doi.org/10.1109/TIT.2005.851738}
  {\bibfield  {journal} {\bibinfo  {journal} {IEEE Trans. Inform. Theory}\
  }\textbf {\bibinfo {volume} {51}},\ \bibinfo {pages} {2812} (\bibinfo {year}
  {2005})}\BibitemShut {NoStop}%
\bibitem [{\citenamefont {Ac\'{\i}n}\ \emph {et~al.}(2005)\citenamefont
  {Ac\'{\i}n}, \citenamefont {Gill},\ and\ \citenamefont
  {Gisin}}]{Acin.etal2005}%
  \BibitemOpen
  \bibfield  {author} {\bibinfo {author} {\bibfnamefont {A.}~\bibnamefont
  {Ac\'{\i}n}}, \bibinfo {author} {\bibfnamefont {R.}~\bibnamefont {Gill}},\
  and\ \bibinfo {author} {\bibfnamefont {N.}~\bibnamefont {Gisin}},\ }\bibfield
   {title} {\bibinfo {title} {Optimal {Bell} tests do not require maximally
  entangled states},\ }\href {https://doi.org/10.1103/PhysRevLett.95.210402}
  {\bibfield  {journal} {\bibinfo  {journal} {Phys. Rev. Lett.}\ }\textbf
  {\bibinfo {volume} {95}},\ \bibinfo {pages} {210402} (\bibinfo {year}
  {2005})}\BibitemShut {NoStop}%
\bibitem [{\citenamefont {Jungnitsch}\ \emph {et~al.}(2010)\citenamefont
  {Jungnitsch}, \citenamefont {Niekamp}, \citenamefont {Kleinmann},
  \citenamefont {G\"uhne}, \citenamefont {Lu}, \citenamefont {Gao},
  \citenamefont {Chen}, \citenamefont {Chen},\ and\ \citenamefont
  {Pan}}]{Jungnitsch.etal2010}%
  \BibitemOpen
  \bibfield  {author} {\bibinfo {author} {\bibfnamefont {B.}~\bibnamefont
  {Jungnitsch}}, \bibinfo {author} {\bibfnamefont {S.}~\bibnamefont {Niekamp}},
  \bibinfo {author} {\bibfnamefont {M.}~\bibnamefont {Kleinmann}}, \bibinfo
  {author} {\bibfnamefont {O.}~\bibnamefont {G\"uhne}}, \bibinfo {author}
  {\bibfnamefont {H.}~\bibnamefont {Lu}}, \bibinfo {author} {\bibfnamefont
  {W.-B.}\ \bibnamefont {Gao}}, \bibinfo {author} {\bibfnamefont {Y.-A.}\
  \bibnamefont {Chen}}, \bibinfo {author} {\bibfnamefont {Z.-B.}\ \bibnamefont
  {Chen}},\ and\ \bibinfo {author} {\bibfnamefont {J.-W.}\ \bibnamefont
  {Pan}},\ }\bibfield  {title} {\bibinfo {title} {Increasing the statistical
  significance of entanglement detection in experiments},\ }\href
  {https://doi.org/10.1103/PhysRevLett.104.210401} {\bibfield  {journal}
  {\bibinfo  {journal} {Phys. Rev. Lett.}\ }\textbf {\bibinfo {volume} {104}},\
  \bibinfo {pages} {210401} (\bibinfo {year} {2010})}\BibitemShut {NoStop}%
\bibitem [{\citenamefont {Jungnitsch}\ \emph {et~al.}(2011)\citenamefont
  {Jungnitsch}, \citenamefont {Moroder},\ and\ \citenamefont
  {G\"uhne}}]{Jungnitsch.etal2011}%
  \BibitemOpen
  \bibfield  {author} {\bibinfo {author} {\bibfnamefont {B.}~\bibnamefont
  {Jungnitsch}}, \bibinfo {author} {\bibfnamefont {T.}~\bibnamefont
  {Moroder}},\ and\ \bibinfo {author} {\bibfnamefont {O.}~\bibnamefont
  {G\"uhne}},\ }\bibfield  {title} {\bibinfo {title} {Taming multiparticle
  entanglement},\ }\href {https://doi.org/10.1103/PhysRevLett.106.190502}
  {\bibfield  {journal} {\bibinfo  {journal} {Phys. Rev. Lett.}\ }\textbf
  {\bibinfo {volume} {106}},\ \bibinfo {pages} {190502} (\bibinfo {year}
  {2011})}\BibitemShut {NoStop}%
\bibitem [{\citenamefont {{Neven}}\ \emph {et~al.}()\citenamefont {{Neven}},
  \citenamefont {{Carrasco}}, \citenamefont {{Vitale}}, \citenamefont
  {{Kokail}}, \citenamefont {{Elben}}, \citenamefont {{Dalmonte}},
  \citenamefont {{Calabrese}}, \citenamefont {{Zoller}}, \citenamefont
  {{Vermersch}}, \citenamefont {{Kueng}},\ and\ \citenamefont
  {{Kraus}}}]{Neven.etal2021}%
  \BibitemOpen
  \bibfield  {author} {\bibinfo {author} {\bibfnamefont {A.}~\bibnamefont
  {{Neven}}}, \bibinfo {author} {\bibfnamefont {J.}~\bibnamefont {{Carrasco}}},
  \bibinfo {author} {\bibfnamefont {V.}~\bibnamefont {{Vitale}}}, \bibinfo
  {author} {\bibfnamefont {C.}~\bibnamefont {{Kokail}}}, \bibinfo {author}
  {\bibfnamefont {A.}~\bibnamefont {{Elben}}}, \bibinfo {author} {\bibfnamefont
  {M.}~\bibnamefont {{Dalmonte}}}, \bibinfo {author} {\bibfnamefont
  {P.}~\bibnamefont {{Calabrese}}}, \bibinfo {author} {\bibfnamefont
  {P.}~\bibnamefont {{Zoller}}}, \bibinfo {author} {\bibfnamefont
  {B.}~\bibnamefont {{Vermersch}}}, \bibinfo {author} {\bibfnamefont
  {R.}~\bibnamefont {{Kueng}}},\ and\ \bibinfo {author} {\bibfnamefont
  {B.}~\bibnamefont {{Kraus}}},\ }\href@noop {} {\bibinfo {title}
  {Symmetry-resolved entanglement detection using partial transpose moments}},\
  \Eprint {https://arxiv.org/abs/2103.07443} {arXiv:2103.07443} \BibitemShut
  {NoStop}%
\bibitem [{\citenamefont {Horn}\ and\ \citenamefont
  {Johnson}(2012)}]{Horn.Johnson2012}%
  \BibitemOpen
  \bibfield  {author} {\bibinfo {author} {\bibfnamefont {R.~A.}\ \bibnamefont
  {Horn}}\ and\ \bibinfo {author} {\bibfnamefont {C.~R.}\ \bibnamefont
  {Johnson}},\ }\href@noop {} {\emph {\bibinfo {title} {Matrix Analysis}}}\
  (\bibinfo  {publisher} {Cambridge University Press, Cambridge},\ \bibinfo
  {year} {2012})\BibitemShut {NoStop}%
\bibitem [{\citenamefont {Werner}(1989)}]{Werner1989}%
  \BibitemOpen
  \bibfield  {author} {\bibinfo {author} {\bibfnamefont {R.~F.}\ \bibnamefont
  {Werner}},\ }\bibfield  {title} {\bibinfo {title} {Quantum states with
  {Einstein}-{Podolsky}-{Rosen} correlations admitting a hidden-variable
  model},\ }\href {https://doi.org/10.1103/PhysRevA.40.4277} {\bibfield
  {journal} {\bibinfo  {journal} {Phys. Rev. A}\ }\textbf {\bibinfo {volume}
  {40}},\ \bibinfo {pages} {4277} (\bibinfo {year} {1989})}\BibitemShut
  {NoStop}%
\bibitem [{\citenamefont {Bhatia}(1997)}]{Bhatia1997}%
  \BibitemOpen
  \bibfield  {author} {\bibinfo {author} {\bibfnamefont {R.}~\bibnamefont
  {Bhatia}},\ }\href@noop {} {\emph {\bibinfo {title} {Matrix Analysis}}}\
  (\bibinfo  {publisher} {Springer-Verlag, New York},\ \bibinfo {year}
  {1997})\BibitemShut {NoStop}%
\bibitem [{\citenamefont {Berry}\ and\ \citenamefont
  {Sanders}(2003)}]{Berry.Sanders2003}%
  \BibitemOpen
  \bibfield  {author} {\bibinfo {author} {\bibfnamefont {D.~W.}\ \bibnamefont
  {Berry}}\ and\ \bibinfo {author} {\bibfnamefont {B.~C.}\ \bibnamefont
  {Sanders}},\ }\bibfield  {title} {\bibinfo {title} {Bounds on general entropy
  measures},\ }\href {https://doi.org/10.1088/0305-4470/36/49/008} {\bibfield
  {journal} {\bibinfo  {journal} {J. Phys. A: Math. Gen.}\ }\textbf {\bibinfo
  {volume} {36}},\ \bibinfo {pages} {12255} (\bibinfo {year}
  {2003})}\BibitemShut {NoStop}%
\bibitem [{\citenamefont {Szyma{\'n}ski}\ \emph {et~al.}(2017)\citenamefont
  {Szyma{\'n}ski}, \citenamefont {Collins}, \citenamefont {Szarek},\ and\
  \citenamefont {{\.Z}yczkowski}}]{Szymanski.etal2017}%
  \BibitemOpen
  \bibfield  {author} {\bibinfo {author} {\bibfnamefont {K.}~\bibnamefont
  {Szyma{\'n}ski}}, \bibinfo {author} {\bibfnamefont {B.}~\bibnamefont
  {Collins}}, \bibinfo {author} {\bibfnamefont {T.}~\bibnamefont {Szarek}},\
  and\ \bibinfo {author} {\bibfnamefont {K.}~\bibnamefont {{\.Z}yczkowski}},\
  }\bibfield  {title} {\bibinfo {title} {Convex set of quantum states with
  positive partial transpose analysed by hit and run algorithm},\ }\href
  {https://doi.org/10.1088/1751-8121/aa70f5} {\bibfield  {journal} {\bibinfo
  {journal} {J. Phys. A: Math. Theor.}\ }\textbf {\bibinfo {volume} {50}},\
  \bibinfo {pages} {255206} (\bibinfo {year} {2017})}\BibitemShut {NoStop}%
\end{thebibliography}%

\end{document}